\documentclass[11pt]{article}
\usepackage{float}
\usepackage{amsmath}
\usepackage{amsthm}
\usepackage{amssymb}
\usepackage{algorithm, algorithmicx, algpseudocode}
\usepackage{graphicx}
\usepackage{cleveref}
\usepackage{color}
\usepackage{supertech}
\usepackage{centernot}
\usepackage{enumitem} 
\usepackage[a4paper,left=1in,right=1in,top=1in,bottom=1in,footskip=.25in]{geometry}
\usepackage{url}
\usepackage{comment}
\usepackage{caption}
\usepackage[numbers,sort]{natbib} 
\usepackage{multibib}
\setlist{nosep}

\newcommand{\pparagraph}[1]{\vspace{0.1in}\noindent{\bf \boldmath #1.}}

\renewcommand{\epsilon}{\varepsilon}

\renewcommand{\sf}[1]{\mathsf{#1}}

\newcommand{\OSAT}{\sf{Oracle\mbox{-}3SAT}}
\newcommand{\cMRIP}{\sf{O(1)}\mbox{-}\sf{MRIP}}
\newcommand{\aMRIP}{\sf{\alpha(n)}\mbox{-}\sf{MRIP}}
\newcommand{\pMRIP}{\sf{poly(n)}\mbox{-}\sf{MRIP}}

\date{}
\title{Rational Proofs with Multiple Provers\thanks{A preliminary version of this paper~\cite{ChenMcSi16} appeared at the 7th Innovations in Theoretical Computer Science (ITCS 2016). This is the full version that contains new results.}
}

\author{
  Jing Chen\thanks{
Computer Science Department, Stony Brook University,
Stony Brook, NY 11794, USA. Email:~\texttt{\{jingchen, shiksingh\}@cs.stonybrook.edu}}
\and
Samuel McCauley\thanks{
IT University of Copenhagen, Rued Langgaards Vej 7, 2300 Copenhagen S, Denmark.
Email:~\texttt{samc@itu.dk}}
\and
Shikha Singh\footnotemark[2]
}

\begin{document}

\maketitle
\sloppy
\begin{abstract}

Interactive proofs (IP) model a world where a verifier delegates computation to an untrustworthy prover, verifying the prover's claims before accepting them.
IP protocols have applications in areas such as verifiable computation outsourcing, computation delegation, cloud computing, etc.
In these applications, the verifier may pay the prover based on the quality of his work.
Rational interactive proofs (RIP), introduced by Azar and Micali (2012),
are an interactive-proof system with payments, in which the prover is {\em rational} rather than untrustworthy---he may lie, but only to increase his payment.
Rational proofs leverage the prover's rationality
to obtain simple and efficient protocols.
Azar and Micali show that RIP=IP(=PSPACE), i.e., 
the set of provable languages stay the same with a single rational prover (compared to classic IP).
They leave the question of whether multiple provers are more powerful than a single prover for rational and classical proofs as an open problem.

In this paper we introduce multi-prover rational interactive proofs (MRIP).
Here, a verifier cross-checks the provers' answers with each other
and pays them according to the messages exchanged. The provers are {\em cooperative} and maximize their {total} expected
payment if and only if the verifier learns the correct answer to the problem.
We further refine the model of MRIP to incorporate {\em utility gaps}, which is the loss in payment suffered by provers
who mislead the verifier to the wrong answer.

We define the class of MRIP protocols with constant, noticeable and negligible utility 
gaps---the payment loss due to a wrong answer is $O(1)$, $1/n^{O(1)}$ and $1/2^{n^{O(1)}}$ respectively, where $n$ is the length of the input.
We give tight characterization for all three MRIP classes.
On the way, we resolve Azar and Micali's open problem---under standard complexity-theoretic assumptions,
MRIP is not only more powerful than RIP,
but also more powerful than MIP (classic multi-prover IP);
and this is true even the utility gap is required to be constant.  We further show that the
full power of each MRIP class can be achieved using only
two provers and three rounds of communication.

\end{abstract}

\section{Introduction} \seclabel{introduction}
Multi-prover interactive proofs (MIP) \cite{ben1988multi} and rational
interactive proofs (RIP)~\cite{azar2012rational} are two important
extensions of classic (single-prover) interactive proof systems \cite{IP, babai1985trading}.
In a multi-prover interactive proof, several computationally unbounded,
potentially dishonest provers interact with a polynomial-time, randomized
verifier.  The provers can pre-agree on a joint strategy to convince the
verifier about the truth of a proposition.  However, once the protocol starts, the
provers cannot communicate with each other.  If the proposition is true, the
verifier should be convinced with probability~1; otherwise the verifier should
reject with probability at least $2/3$.
It is well known that $\sf{MIP} = \sf{NEXP}$~\cite{babai1991non}, which demonstrates the power
of multiple provers compared to single-prover interactive proofs ---recall that
$\sf{IP} =\sf{PSPACE}$ \cite{Shamir92,lund1992algebraic}.

Rational interactive proofs \cite{azar2012rational} are
interactive proofs in which the verifier makes a payment to the prover at the end
of the protocol. The prover is assumed to be {\em rational}: that is, he only
acts in ways that maximize his expected payment.
Thus, unlike classic interactive
proofs, in rational proofs the prover does not care whether the verifier is convinced or not.
Instead, rational proofs ensure that the prover maximizes the expected payment if and only if
the verifier learns the truth of the proposition.
In~\cite{azar2012rational} Azar and Micali introduce rational proofs and show that while rational proofs are no more powerful
than classic interactive proofs in terms of the class of provable languages (i.e.,  $\sf{RIP} = \sf{PSPACE}$), the protocols are
simpler and more efficient. They have left the power (and the model) of multi-prover rational proofs as an open problem.

Meanwhile, many real-world computation-outsourcing applications have ingredients of both MIP and RIP:
the verifier pays a team of provers based on their responses.  For
example, in Internet marketplaces such as Mechanical
Turk~\cite{turk} and Proof Market~\cite{proofmarket}, the requesters
(verifiers) post labor-intensive tasks on the website along with a monetary
compensation they are willing to pay. The providers (provers) accept these
offers and perform the job.  In these marketplaces and other crowdsourcing
scenarios~\cite{von2008designing}, correctness is often ensured by verifying one
provider's answers against another~\cite{von2004labeling,effective-turk}.
Thus, the providers implicitly work as a team---their answers need to match,
even though they are likely to not know each other and cannot communicate with
each other~\cite{kittur2010crowdsourcing}. While these applications differ from interactive proofs in many ways,
they motivate the study of a proof system with multiple cooperative rational provers.

Inspired by these applications and previous theoretical work, we introduce {\em multi-prover rational
interactive proofs}, which combine elements of rational proofs and classical multi-prover interactive
proofs.
%
%
In this paper we focus on the following question: what computation problems can be solved by a team of
rational workers who get paid based on
the joint-correctness of their answers and cannot communicate with each other?
The main contribution of this paper is to fully characterize the power of such computation-outsourcing systems,
under different requirements on the payment loss suffered when the provers solve the problem incorrectly.
Our model is formally defined in Section \ref{sec:mrip}; we briefly summarize it below.


%
%
%

\pparagraph{Cooperative Multi-Prover Rational Proofs}
In a cooperative multi-prover rational interactive proof, polynomially many computationally-unbounded
provers communicate with a polynomial-time randomized verifier, where the verifier wants to
decide the membership of an input string in a language.  The provers can
pre-agree on how they plan to respond to the verifier's messages, but they
cannot communicate with each other once the protocol begins.  At the end of the
protocol, the verifier outputs the answer and computes a total payment for the
provers, based on the input, his own randomness, and the messages exchanged.

A protocol is an MRIP protocol if any strategy of the provers that maximizes
their expected payment leads the verifier to the correct answer.  The class of
languages having such protocols is denoted by $\sf{MRIP}$.
Note that classical multi-prover interactive proofs are robust against
arbitrary malicious provers;
MRIP protocols instead require provers to be rational---a reasonable
requirement in a ``mercantile world''~\cite{azar2012rational}.

\pparagraph{Distribution of Payments}
In classical MIP protocols, the provers work cooperatively
to convince the verifier of the truth of a proposition, and their goal is to maximize
the verifier's acceptance probability. Similarly,
the rational provers in MRIP protocols work cooperatively to maximize the total payment
received from the verifier.
Any pre-specified way of distributing this payment among them is allowed, as long
as it does not depend on the transcript of the protocol (i.e., the messages exchanged, the coins flipped, and the amount of the payment).
For instance, the division of the payment can be pre-determined by the provers themselves based on the amount of work each prover must perform,
or it can be pre-determined by the verifier based on the reputation of each prover in a marketplace.
Unbalanced divisions are allowed: for example, one prover may receive half
of the total payment, while the others split the remaining evenly.
We will ignore the choice of division in our model and protocols, as it does not affect the provers' decisions when choosing
their strategy.


\pparagraph{Utility Gaps}
Rational proofs assume that the provers always act to maximize
their payment. However, how much do they lose by lying? If the payment loss is small, a prover may very well ``get lazy'' and simply return a default answer without
performing any computation.
Although the classic notion of rationality in game theory
requires a player to always choose the best strategy to maximize his utility,
the notion of bounded rationality has also been studied \cite{rubinstein1998modeling, conlisk1996bounded}.

The notion of {\em utility gap}
measures the payment or utility loss incurred by a deviating prover.
A deviating prover may (a) deviate slightly from the truthful protocol but still lead the verifier to the correct answer or
(b) deviate and mislead the verifier to an incorrect answer. Azar and Micali~\cite{azar2013super} introduce utility gaps
by demanding their protocols be robust against provers of
type (a)---any deviation from the prescribed strategy
results in a significant decrease in the payment.
This ideal requirement on utility gaps is too strong: even the protocol
in~\cite{azar2013super} fails to satisfy it~\cite{guo2014rational}.

In this work, we
consider multi-prover rational proofs robust against provers of type (b), i.e.,
the provers may send some incorrect messages and only incur a small
payment loss, but if they mislead the verifier to the wrong answer to the membership question
of the input string, then the provers must
suffer a significant loss in the payment.
Such deviations were also considered in~\cite{guo2014rational}, but for single-prover protocols and with a slightly different notion
of utility gaps.

We strengthen our model by considering
MRIP protocols
with {\em constant} as well as {\em noticeable} (i.e. {\em polynomial}) utility gaps,
where the payment loss suffered by the provers on reporting the incorrect answer is at least $1/k$ and $1/n^k$ respectively,
where $k$ is a constant and $n$ is the length of the input string.
We say an MRIP protocol has a {\em negligible} (or {\em exponential}) utility gap if the
payment loss is at least $1/2^{n^k}$.
Any MRIP protocol has at least a negligible utility gap, because the rewards are generated by a polynomial-time verifier.

\pparagraph{Complexity Classes With Oracle Queries}  Our
characterizations of MRIP protocols are closely related to complexity classes with oracle queries.
In particular, let $\sf{EXP^{||NP}}$ be the class of languages decidable
by an exponential-time Turing machine with non-adaptive access to an $\sf{NP}$ oracle. Note that in this case, the queries may be exponentially long.
Non-adaptive access means that all queries must be decided before any one query is made; they may not depend on each other.
Similar classes, such as $\sf{P^{||NEXP}}$, are defined analogously.
In some cases we consider complexity classes where the number of oracle queries is limited.
For example, $\sf{P^{||NEXP[O(1)]}}$ is the class of languages decidable by a polynomial-time
Turing machine which can make $O(1)$ non-adaptive queries to an $\sf{NEXP}$ oracle.

Many of these classes have been studied previously; see \secref{related}.

\subsection{Main Results}
We now present our main results and discuss several interesting aspects of our model.

\pparagraph{The Power of Multi-Prover Rational Proofs}
We denote the classes of {MRIP} protocols
with constant, polynomial and exponential utility gaps as $\cMRIP$,
$\pMRIP$ and $\sf{MRIP}$ respectively. By definition, $\cMRIP \subseteq \pMRIP \subseteq \sf{MRIP}$.

In this work, we fully characterize the computation power of all three MRIP classes. 


\begin{theorem}\thmlabel{constantchar}
$\cMRIP = \sf{P^{||NEXP[O(1)]}}$.
\end{theorem}

That is, a language has an MRIP protocol with constant utility gap if and only if it can be decided by a polynomial-time Turing machine that makes a constant number of non-adaptive queries
to an $\sf{NEXP}$ oracle.

\thmref{constantchar} implies that $\cMRIP$ contains both $\sf{NEXP}$ and $\sf{coNEXP}$. That is,
multi-prover rational proofs with even {\em constant} utility gaps are
strictly more powerful
than single-prover rational proofs,
assuming $\sf{PSPACE \neq NEXP}$. Furthermore, multi-prover rational proofs (even with constant utility gaps)
are strictly more powerful than classical
multi-prover interactive proofs,
assuming $\sf{NEXP}\neq \sf{coNEXP}$.
The relationship between rational and classical interactive proof systems is illustrated in Figure \ref{fig:classes}.
\begin{figure}[h]
     \begin{center}
\resizebox{0.4\textwidth}{!}{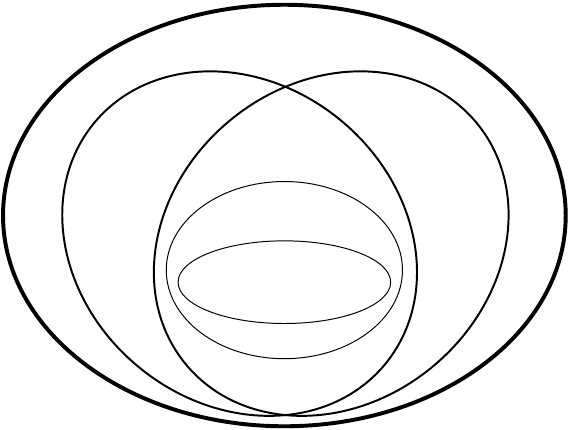}
    \end{center}
\caption{The computation power of rational and classical interactive proof systems.
    Note that it is widely believed that $\sf{PSPACE} \neq \sf{EXP}$, $\sf{EXP} \neq \sf{NEXP}$, and $\sf{NEXP} \neq \sf{coNEXP}$.
}
\label{fig:classes}
\end{figure}

\begin{theorem}\thmlabel{polychar}
$\pMRIP = \sf{P^{||NEXP}}$.
\end{theorem}

That is, a language has an MRIP protocol with polynomial utility
gap if and only if it can be decided by a polynomial-time Turing machine with non-adaptive access
to an $\sf{NEXP}$ oracle.

\begin{theorem}\thmlabel{expchar}
$\mathsf{MRIP = EXP^{||NP}}.$
\end{theorem}

That is,  a language has an MRIP protocol (with exponential utility gap)
 if and only if it can be decided by an exponential-time Turing machine with non-adaptive access
to an $\sf{NP}$ oracle.


\smallskip

We give MRIP protocols for $\sf{NEXP}$, which are used as a building block in our proofs.
To prove \thmref{constantchar} and \thmref{polychar}, we establish 
a general reduction between the utility gap of MRIP protocols
and the query complexity of oracle Turing machines. This reduction may be of independent interest 
when studying the relationship between these two computation models.

Finally, to prove \thmref{expchar}, we introduce another complexity class as an intermediate step,
and use its circuit characterization to construct the corresponding MRIP protocol.
Similar circuit based characterization is also used by Azar and Micali in~\cite{azar2013super}, but 
but their technique results in
an exponential blow-up in the number of messages when applied directly to our case. We use multiple provers
to avoid this communication blow up; see~\secref{mripchar}.

%
%


\pparagraph{MRIP with Two Provers and Constant Number of Rounds}\label{subsubsec:23}
While we allow polynomially many provers and rounds in MRIP, how many provers and rounds are really needed to capture the full power
of the system? In real-world applications, protocols with few provers and rounds are desirable, as it may be hard for the
verifier to recruit a large number of provers or to retain the provers for a long period of time to execute many rounds.

Under the classic model of interactive proofs, it is well known that any MIP protocol
can be simulated using only two
provers and one round of communication between the provers and the
verifier~\cite{feige1992two}.
In this work, we prove analogous results for all three of our MRIP classes.

Specifically, let $\sf{MRIP}[p(n), k(n), t(n)]$ denote the class of languages that have
MRIP protocols with $p(n)$ provers, $k(n)$ rounds, and $1/t(n)$ utility gap, and let
$\mbox{poly}(n)$ and $\mbox{exp}(n)$ denote the class of polynomial and exponential functions respectively, where $n$ is the input length. 

%


\begin{theorem}\thmlabel{mrip-const-23}
$\cMRIP = \sf{MRIP}[2,3, O(1)]$.
\end{theorem}

\begin{theorem}\thmlabel{mrip-poly-23}
$\pMRIP = \sf{MRIP}[2,3, \mbox{\em poly}(n)]$.
\end{theorem}

\begin{theorem}\thmlabel{mrip-23}
$\sf{MRIP} = \sf{MRIP}[2,3, \mbox{\em exp}(n)]$.
\end{theorem}

That is, any MRIP protocol using polynomially many provers and polynomially many rounds that has a constant, polynomial, or exponential utility gap can be simulated by a 2-prover 3-round MRIP protocol that retains the utility gap.
Our analysis for $\cMRIP$ and $\pMRIP$
relies on a different approach than that for $\sf{MRIP}$, and all three theorems are based on our characterizations for the corresponding general MRIP classes.



\smallskip

It is worth pointing out that we count the number of rounds in a protocol differently from classic IP and MIP protocols.
In the classic protocols, the number of rounds
is the number of {\em pairs} of back-and-forth
interactions (see, e.g., \cite{feige1992two}); while in our protocols it is the total number of interactions---that is,
the provers' messages and the verifier's messages are considered as different rounds.
An odd number of rounds
is an intrinsic property of multi-prover rational proofs, as an MRIP protocol by default starts with the provers
reporting the {\em answer bit} to the verifier (see~\secref{mrip} for details).
Thus, the 3-round protocols consist of the first ``answer bit round'',
followed by a single back-and-forth exchange corresponding to a single round in IP or  MIP.
Indeed, any non-trivial MRIP protocol---that is, any MRIP protocol that cannot be simulated by a single prover---requires
at least three rounds. Thus, three rounds are optimal and
all the theorems above are tight.


Finally, we note that the power of MRIP protocols remains the same even when it is restricted to constant number of rounds,
while the power of RIP protocols decreases. In particular, Azar and Micali~\cite{azar2012rational} show that
the class of languages having constant-round single-prover rational proofs is exactly the counting hierarchy,
while $\sf{RIP = PSPACE}$.
This difference between $\sf{MRIP}$ and $\sf{RIP}$ is analogous to the difference between $\sf{MIP}$ and $\sf{IP}$.

    \subsection{Related Work}\seclabel{related}

\pparagraph{Interactive Proofs}
First introduced by Goldwasser et al.~\cite{IP} and in a different
form by Babai and Moran~\cite{babai1985trading}, {\em interactive proofs} (IP) have
been extensively studied in the literature~\cite{goldreich1991proofs, babai1991non,
ben1988multi, babai1988arthur,
fortnow1994power,fortnow1988there, goldwasser1986private} and fully
characterized by the well known result, $\sf{IP}=\sf{PSPACE}$~\cite{Shamir92,lund1992algebraic}.
Ben-Or et al.~\cite{ben1988multi} introduced multi-prover interactive proofs
(MIP), which has been shown to be exactly $\sf{NEXP}$~\cite{babai1991non}.
In fact, two provers and one round are sufficient to achieve the full power of multi-prover interactive proofs:
that is, $\sf{NEXP = MIP(2,1)}$~\cite{feige1992two}.

Introduced by Goldwasser et al.~\cite{goldwasser2008delegating}, {\em computation delegations} are interactive proofs where the provers are also computationally bounded.
These protocols have been studied by many ever since; see, for example,~\cite{rothblum2013interactive,gur2015non,
kalai2015arguments, canetti2013refereed,canetti2012two, reingold2016constant}. 
Recently, interactive proofs have also been studied in {\em streaming} settings~\cite{cormode2012practical,chakrabarti2015verifiable,
cormode2011verifying,daruki2015streaming}.

\pparagraph{Rational Proofs}
Azar and Micali~\cite{azar2012rational} first introduced {\em rational interactive proofs} (RIP)
and used scoring rules to construct simple and efficient (single-prover) RIP protocols.
In~\cite{azar2013super}, the same authors designed super-efficient rational proofs, where the verifier runs in logarithmic time and the communication complexity is sublinear.
Guo et al.~\cite{guo2014rational} considered rational arguments
for a computationally bounded prover and a sublinear verifier. In~\cite{guo2016rational},
the same authors constructed
rational arguments for all languages in $\sf{P}$.
Moreover, Campanelli and Rosario~\cite{campanelli2015sequentially} study sequentially composable rational proofs and
Zhang and Blanton~\cite{zhang2014efficient} design protocols to outsource matrix multiplications to a rational cloud.

\pparagraph{Refereed Games} Game-theoretic
characterization of complexity classes has been studied in the form
of {\em refereed games}~\cite{chandra1976alternation, feige1990noisy,
feige1997making, feige1992multi, reif1984complexity,
feigenbaum1995game,koller1992complexity}.
They are interactive proofs consisting of
two competing provers.
One of them is always honest and tries to convince the verifier of the
membership (or non-membership) of an input string in a language; the other 
is always dishonest and tries to mislead the verifier.
Chandra and Stockmeyer~\cite{chandra1976alternation} show
that any language in
$\sf{PSPACE}$ is refereeable by a game of perfect information,
and
Feige and Kilian~\cite{feige1997making} show that this is tight for single-round refereed games
and that the class of
languages with polynomial-round refereed games is exactly $\sf{EXP}$.

Feigenbaum et al.~\cite{feigenbaum1995game} show that any language in $\sf{EXP^{NP}}$ can be simulated as a zero-sum refereed game between two
computationally unbounded provers with {\em imperfect recall}.
Note that imperfect recall is a strong assumption and makes the
computationally unbounded provers essentially act as oracles.
By contrast, MRIP protocols have cooperative provers with {\em imperfect information} ---since a prover does not see the
messages exchanged between the verifier and the other provers--- and {\em perfect recall}
---since a prover remembers the
history of messages exchanged between himself and the verifier.  Note that imperfect
information is necessary for multi-prover protocols:
if all provers can see all messages exchanged in the protocol, then
the model degenerates to a single-prover case.
Moreover, perfect recall gives the provers the ability to cheat adaptively across rounds.
To the best of our knowledge, $\sf{MRIP}$ gives the first game-theoretic characterization of the class $\sf{EXP^{||NP}}$.

\pparagraph{Query Complexity and Related Complexity Classes}
The query complexity of oracle Turing machines
has been widely studied in the literature~\cite{beigel1991bounded,wagner1990bounded,buhrman1999quantum}.
As shown by our work, the computation power of multi-prover rational proofs is closely related to the query complexity of oracle Turing machines.
Finally, it is worth pointing out that $\sf{EXP^{NP}}$ is an important complexity class in
the study of circuit lower bounds~\cite{williams2014nonuniform}.
It would be interesting to see if the related class $\sf{EXP^{||NP}}$ emerges in similar contexts.

\section{Multi-Prover Rational Interactive Proofs}
\seclabel{mrip}

In this section, we first define multi-prover rational interactive proofs (MRIP) in general,
and then strengthen the model by imposing proper utility gaps.

\subsection{Basic Notation and Definitions}
Let $L$ be a language, $x$ a string whose membership in $L$ is to be decided, and $n= |x|$. 
An {\em interactive protocol} is a pair $(V, \vec{P})$, where
$V$ is the {\em verifier} and
$\vec{P} = (P_1,\ldots,P_{p(n)})$ is the vector of {\em provers}, and $p(n)$ a polynomial in $n$.
The verifier runs in polynomial time and flips private coins,
whereas each prover $P_i$ is computationally unbounded.
The verifier and provers know $x$.
The verifier can communicate with each prover privately, but no two provers can communicate with each other.
In a \emph{round}, either each prover sends a message to the verifier, or the
verifier sends a message to each prover, and these two cases alternate.  Without loss of generality, we assume the first round of messages are sent by the provers,
and the first bit sent by $P_1$, denoted by $c$, indicates whether $x\in L$ (corresponding to $c=1$) or not (corresponding to $c=0$).

The length of each message and the number
of rounds are polynomial in $n$.
Let $k(n)$ be the number of rounds and $r$ be the random string used by $V$.
For each $j\in \{1, 2, \dots, k(n)\}$, let $m_{ij}$ be the message exchanged between $V$ and $P_i$ in round $j$.
In particular, the first bit of $m_{11}$ is $c$.
The transcript that each prover $P_i$ has seen at the beginning of each round $j$
is $(m_{i1}, m_{i2}, \dots, m_{i(j-1)})$.
Let $\vec{m}$ be the vector of all messages exchanged in the protocol. By definition, $\vec{m}$ is a random variable depending on $r$.

At the end of the communication, the verifier
evaluates the total payment to the provers, given by
a payment function $R$ on $x$, $r$, and $\vec{m}$.
We restrict $R(x, r, \vec{m})\in [-1, 1]$ for convenience.
Of course, the payment can be shifted so that it is non-negative---that is, the provers do not lose money.
We use both positive and negative payments to better
reflect the intuition behind our protocols: the former are rewards while the latter are punishments.
The protocol followed by $V$, including the payment function~$R$, is public knowledge.

The verifier outputs $c$ as the answer for the membership of $x$ in $L$---that is, $V$ does not check the provers' answer.
This requirement for the verifier does not change the set of languages that have multi-prover rational interactive proofs; however, it simplifies our later discussion of utility gaps
 (i.e., the payment loss incurred by  provers that report the wrong answer).

\subsection{Cooperative Multi-Prover Rational Proofs}
Each prover $P_i$ can choose a {\em strategy} $s_{ij}:\{0,1\}^*\rightarrow\{0,1\}^*$
for each round $j$, which maps the transcript
he has seen up until the beginning of round $j$
to the message he sends in round $j$.
Note that $P_i$ does not send any
message when $j$ is even;  in this case $s_{ij}$ can be treated as a constant function.
Let $s_i = (s_{i1},\ldots, s_{ik(n)})$ be the strategy vector of $P_i$
and $s = (s_1,\dots, s_{p(n)})$ be the strategy profile of the provers.
Given any input $x$, randomness $r$ and strategy profile $s$, we may write the vector $\vec{m}$ of messages exchanged in the protocol
more explicitly as $(V, \vec{P})(x, r, s)$.

The provers are {\em cooperative} and jointly act to maximize the total expected payment received from the verifier.
Note that this is equivalent to each prover maximizing his own expected payment when each $P_i$
receives a pre-specified fraction $\gamma_i$ of the payment, where
$\sum_{i=1}^{p(n)} \gamma_i =1$ and $\gamma_i$ may depend on $x$ but not on $r$ and $\vec{m}$.

Thus, before the protocol starts, the provers pre-agree on a strategy profile $s$ that maximizes
\[
    u_{(V, \vec{P})}(s; x) \triangleq \mathop{\scalebox{1.25}{$\mathbb{E}$}}_r \left[R\left(x, r, (V, \vec{P})(x, r, s)\right)\right].
\]
When $(V, \vec{P})$ and $x$ are clear from the context, we write $u(s)$ for $u_{(V, \vec{P})}(s; x)$.
We define multi-prover rational interactive proofs as follows.

\begin{definition}[$\sf{MRIP}$]
    \label{def:mrip}
For any language $L$, an interactive protocol $(V, \vec{P})$ is a {\em multi-prover rational interactive proof (MRIP)} protocol for $L$ if,
for any $x\in \{0, 1\}^*$ and any strategy profile $s$ of the provers
such that $u(s) = \max_{s'} u(s')$, $c =1$ if and only if $x\in L$.
We denote the class of languages that have MRIP protocols by $\sf{MRIP}$.
\end{definition}

This definition immediately leads to the following property.

\begin{lemma}\label{lem:closed}
$\sf{MRIP}$ is closed under complement.
\end{lemma}

\begin{proof}
Consider a language $L \in \sf{MRIP}$. Let $(V, \vec{P})$ be the MRIP protocol for $L$, and $R$ the payment function used by $V$.
We construct an MRIP protocol $(V', \vec{P})$ for $\overline{L}$ as follows.

\begin{itemize}[noitemsep,nolistsep,leftmargin=*]
\item
After receiving message $m'_{11}$ from $P_1$, $V'$ flips the first bit. Denote the new message by $m_{11}$.

\item
$V'$ runs $V$ to compute the messages he should send in each round, except that $m'_{11}$ is replaced by $m_{11}$ in the input to $V$.
Let $\vec{m}'$ be the vector of messages exchanged between $V'$ and $\vec{P}$.

\item
At the end of the communication,
$V'$ computes a payment function $R'$: for any $x, r$, and $\vec{m}'$, $R'(x, r, \vec{m}') = R(x, r, \vec{m})$, where $\vec{m}$ is $\vec{m}'$ with
$m'_{11}$ replaced by $m_{11}$.

\item
$V'$ outputs the first bit sent by $P_1$.
\end{itemize}

To see why this is an MRIP protocol for $\overline{L}$, for each strategy profile $s$ of the provers in the protocol $(V, \vec{P})$, consider the following strategy profile $s'$ in the protocol $(V', \vec{P})$.
\begin{itemize}[noitemsep,nolistsep,leftmargin=*]

\item
$s'_i = s_i$ for each $i\neq 1$.

\item
In round 1, $s'_1$ outputs the same message as $s_1$, except that the first bit is flipped.

\item
For any odd $j>1$ and any transcript $m'_1$ for $P_1$ at the beginning of round $j$, $s'_1(m'_1)$ is the same as $s_1(m_1)$, where $m_1$ is $m'_1$ with the first bit flipped.
\end{itemize}

By induction,
for any $x$ and $r$, $(V', \vec{P})(x, r, s')$ is the same as $(V, \vec{P})(x, r, s)$
except the first bit. Thus $R'(x, r, (V', \vec{P})(x, r, s')) = R(x, r, (V, \vec{P})(x, r, s))$, which implies
$
u_{(V', \vec{P})}(s'; x) = u_{(V, \vec{P})}(s; x)$.
Since the mapping from $s$ to $s'$ is a bijection,
if we arbitrarily fix a strategy profile $s'$ that maximizes $u_{(V', \vec{P})}(s'; x)$,
the corresponding strategy profile $s$ maximizes $u_{(V, \vec{P})}(s; x)$.
By definition, $x\in L$ if and only if the first bit sent by $s_1$ is 1;
thus, $x\in \overline{L}$ if and only if the first bit sent by $s'_1$ is 1.
Therefore $(V', \vec{P})$ is an MRIP protocol for $\overline{L}$.
\end{proof}

Note that the MRIP protocols for $\overline{L}$ and $L$ have the same number of provers and the same number of rounds.
Moreover, recall that (assuming $\sf{NEXP}\neq\sf{coNEXP}$) the class of languages having classical multi-prover interactive proofs is not closed under
complement. Thus multi-prover rational proofs are already different from classical ones.

\subsection{MRIP Protocols with Constant and Polynomial Utility Gaps}\seclabel{gap}

In the MRIP model defined so far, the provers are sensitive to
arbitrarily small losses in the payment. That is, the provers choose~$s$ to just maximize their expected payment---the amount
they lose if they use a suboptimal strategy is irrelevant.

In \cite{azar2013super}, Azar and Micali strengthen the RIP model by requiring that the prover deviating from the
optimal strategy suffers a non-negligible loss in the payment. This loss is
demanded for {\em any} deviation,  not just for reporting an incorrect answer.
Formally, let $s$ be an optimal strategy
and $s'$ a suboptimal strategy
of the prover $P$. Then the {\em ideal} utility gap
requires that $u(s) - u(s') > 1/\alpha(n)$, where $\alpha(n)$ is constant
or polynomial in $n$.
Although an ideal utility gap
strongly guarantees that the prover uses his optimal strategy,
as pointed out by \cite{guo2014rational} such a utility gap appears to be too strong to hold for many meaningful protocols, even the ones in~\cite{azar2013super}.

%
In~\cite{guo2014rational}, Guo et al. define a weaker notion of utility gap and impose it on rational arguments rather than rational proofs.
They require that a noticeable deviation leads to a noticeable loss: if under a strategy $s'$ of the prover, the probability for the verifier to output the correct answer
is noticeably smaller than 1, then the expected payment to the prover under $s'$ is also noticeably smaller
than the optimal expected payment.

Our notion of utility gaps is slightly different,
and we require the provers' strategies that report the membership of the input incorrectly
suffer a noticeable loss in the payment.
Any MRIP protocol with our notion of utility gaps satisfy the notion of~\cite{guo2014rational}, but not vice-versa. 

\begin{definition}[Utility Gap]\label{def:rewardgap}
Let $L$ be a language in $\sf{MRIP}$, $(V, \vec{P})$ an
MRIP protocol for $L$, and $\alpha(n) \ge 0$. We say that $(V, \vec{P})$ has an {\em $\alpha(n)$-utility gap} if
for any input $x$ with $|x|=n$, any strategy profile
$s$ of $\vec{P}$ that maximizes the expected payment, and
 any other strategy profile $s'$, where the answer bit $c'$ under $s'$ does not match
the answer bit $c$ under $s$, i.e., $c'\neq c$, then
\[u(s) - u(s') > \frac{1}{\alpha(n)}.
\]
\end{definition}


We denote the class of languages that have an MRIP protocol with constant utility gap
by  $\cMRIP$,
and the class of languages that have an MRIP protocol with polynomial (or noticeable) utility gap
by $\pMRIP$. 
Specifically, 
$\pMRIP$ is the union of MRIP classes with $\alpha(n)$ utility gap,
where $\alpha(n)$ is a polynomial in $n$. $\cMRIP$ is defined analogously.

\pparagraph{Remark}
Since utility gap scales naturally with the payment,
it is important to maintain a fixed budget so as to study them in a consistent way. Otherwise, a polynomial utility gap
under a constant budget can be interpreted as a constant utility gap under a sufficiently-large polynomial budget. 
Thus, we maintain a constant budget and the payment is always in $[-1, 1]$.

Following~\defref{rewardgap}, it is not hard to see that the MRIP protocol for $\overline{L}$ in the proof of Lemma~\ref{lem:closed}
has the same utility gap as the one for $L$. Thus we immediately have the following.

\begin{corollary}\label{col:gap_complement}
$\cMRIP$ and $\pMRIP$ are both closed under complement.
\end{corollary}



%

\section{Warm Up: MRIP Protocols for $\sf{NEXP}$}\label{sec:nexp}
To demonstrate the power of multi-prover rational proofs, 
we start by constructing two different MRIP protocols for $\sf{NEXP}$, the class of languages
decidable by exponential-time non-deterministic Turing machines.



\subsection{A Constant-Gap MRIP Protocol for $\sf{NEXP}$ Based on MIP}\seclabel{nexp-using-mip}

First, we show that $\cMRIP$ contains $\sf{NEXP}$.
We construct the desired MRIP protocol using an MIP protocol 
as a blackbox.
Existing MIP protocols (see, e.g., \cite{babai1991non, feige1992two}) for a language $L\in \sf{NEXP}$ first reduce~$L$ to the $\sf{NEXP}$-complete problem $\OSAT$,
and then run an MIP protocol for $\OSAT$. For completeness, we recall the definition of $\OSAT$ below.

\begin{definition} [$\OSAT$ \cite{babai1991non}] Let $B$ be a 3-CNF of
$r+3s + 3$ variables.
A Boolean function $A:
\{0,1\}^s\rightarrow \{0,1\}$ is a {\em 3-satisfying oracle} for $B$ if
$B(w,A(b_1),A(b_2),A(b_3))$ is satisfied for all binary strings $w$ of length $r + 3s$, where $b_1b_2b_3$ are the
last $3s$ bits of $w$.   The $\OSAT$ problem is to decide, for a given $B$, whether there
is a 3-satisfying oracle for it.
\end{definition}


Below we prove that any language $L\in \sf{NEXP}$ has a 2-prover 3-round MRIP protocol with constant utility gap.

\begin{lemma}\label{lem:nexp-mip}
$\sf{NEXP} \subseteq \sf{MRIP} [2, 3, O(1)].$
\end{lemma}

\begin{proof}
The desired MRIP protocol $(V, \vec{P})$ is defined in \figref{simple-nexp}.

\begin{figure}[phtb]
\centering
\fbox{
\begin{minipage}{0.96\textwidth}
{\normalsize
\vspace{0.5ex}
\noindent{}For any input string $x$, $(V, \vec{P})$ works as follows:
\begin{enumerate}[leftmargin=15pt]
\item $P_1$ sends a bit $c\in \{0,1\}$ to $V$.  $V$ outputs $c$ at the end of the protocol.

\item\label{simple-2} If $c=0$, then the protocol ends and the payment given to the provers is $R = 1/2$;

\item\label{simple-3}
Otherwise, $V$ and $\vec{P}$ run an MIP protocol for proving $x \in L$.
If the verifier accepts then $R = 1$; else, $R = 0$.
\end{enumerate}
}
\end{minipage}
}
\caption{A simple MRIP protocol for $\sf{NEXP}$.}
\figlabel{simple-nexp}
\end{figure}

The 2-prover 3-round MRIP protocol is obtained by running the MIP protocol in~\cite{feige1992two}.
Without loss of generality, let the MIP protocol have completeness 1 and soundness $1/3$.
That is, the verifier accepts every $x\in L$ with probability 1, and every $x\notin L$
with probability at most $1/3$. We show that $V$ outputs~$1$ if and only if $x \in L$.

For any $x \in L$, if the provers send $c=1$ and execute the MIP protocol with $V$, then the payment is $R =1$ because $V$ accepts with probability 1.%
\footnote{If the MIP protocol does not have perfect completeness and accepts $x$ with probability at least $2/3$, then the expected payment is at least $2/3$. This does not affect the correctness of our MRIP protocol.}
If they send $c=0$, then the payment is $R =1/2<1$. 

For any $x\not\in L$, if the provers send $c=1$ and run the MIP protocol,
then the probability that $V$ accepts is at most $1/3$ and
the expected payment is at most $1/3$.
If they send $c=0$, then the payment is $1/2>1/3$. 

Thus, $V$ outputs 1 iff $x\in L$, and $(V, \vec{P})$ is an MRIP protocol for $L$.
Since the provers' payment loss when sending the wrong answer bit is at least
 $1/6$, $(V, \vec{P})$ has $O(1)$ utility gap.
\end{proof}

Combining Corollary \ref{col:gap_complement} and Lemma \ref{lem:nexp-mip}, we have the following.

\begin{corollary} 
$\sf{coNEXP} \subseteq \sf{MRIP} [2, 3, O(1)].$
\end{corollary}


\pparagraph{Remarks}
Three rounds of interaction is the best possible for any non-trivial MRIP protocol with at least two provers, because $P_1$ always sends the answer $c$ in the first round. In particular, if the protocol has only two rounds,
then the last round consists of the verifier sending messages to the provers and can be eliminated. A single-round MRIP protocol 
degenerates into a single-prover rational proof protocol, since the provers can pre-agree on the messages.


The constant utility gap in our MRIP protocol
comes from the constant soundness gap of classical MIP protocols ---that is, the gap between the accepting probability for $x\in L$ and $x\notin L$.
Using the same construction, any classical interactive proof protocol can be converted into an MRIP protocol
where the utility gap is a constant fraction of the original soundness gap.
However, as we show in Section \ref{sec:const-polygap}, this is not the only way to obtain desirable utility gaps.


\subsection{An MRIP Protocol for $\sf{NEXP}$ Based on Scoring Rules}\label{sec:nexp-scoring}

Although our protocol above is conceptually simple,
its implementation relies on MIP protocols, which are technically complex even after the reduction to $\OSAT$.
We now construct an MRIP protocol for any language in $\sf{NEXP}$ without relying on
MIP protocols. Instead, we use a {\em proper scoring rule} to compute the payment for the provers, 
so as to incentivize them to report the correct answer.
However, the way we use the scoring rule is highly non-standard and differs from {\em all} previous uses of scoring rules (including those in rational proofs~\cite{azar2012rational, azar2013super, guo2014rational}).
Let us first recall the notion of proper scoring rules and
{\em Brier's scoring rule}~\cite{brier1950verification} in particular.

\pparagraph{Proper Scoring Rules} Scoring rules are tools to assess the quality
of a probabilistic forecast. It assigns a numerical score (that is, a payment
to the forecaster) to the forecast based on the predicted distribution and the sample
that materializes.  More precisely, given any probability space $\Sigma$,
letting $\Delta(\Sigma)$ be the set of probability distributions over $\Sigma$,
a {\em scoring rule} is a function from $\Delta(\Sigma)\times \Sigma$ to
${\mathbb R}$, the set of reals.  A scoring rule $S$ is {\em proper} if, for
any distribution $D$ over $\Sigma$ and distribution $D' \neq D$, we have \[
\sum_{\omega\in \Sigma} D(\omega)S(D,\omega) \geq \sum_{\omega\in \Sigma}
D(\omega)S(D',\omega), \] where $D(\omega)$ is the probability that $\omega$ is
drawn from $D$.  A scoring rule $S$ is {\em strictly proper} if the above
inequality is strict.
Strictly proper scoring rules are useful because 
a forecaster maximizes his expected score (i.e. payment) 
by reporting the true distribution $D$.  See
\cite{gneiting2007strictly} for a comprehensive survey on scoring rules.

\pparagraph{Brier's Scoring Rule} This classic
scoring rule, which we abbreviate to $\sf{BSR}$, is defined as follows: for any
distribution $D$ and $\omega \in \Sigma$,
\[ \sf{BSR}(D,\omega) = 2D(\omega) - \sum_{\omega\in\Sigma} D(\omega)^2 - 1.
\]
It is well known that $\sf{BSR}$ is strictly proper.

$\sf{BSR}$ requires the computation of $\sum_{\omega\in\Sigma}
D(\omega)^2$, which can be hard when $|\Sigma|$ is large. However, as in
\cite{azar2012rational} and \cite{guo2014rational}, in this work we only consider
$\Sigma = \{0,1\}$.

$\sf{BSR}$ has range $[-2,0]$, but it can be easily
shifted and scaled so that (1) the range is non-negative and bounded, and (2)
the resulting scoring rule is still strictly proper.  In particular, we
add 2 to the classical BSR score when using it, so as to satisfy these requirements.

\smallskip

Next, we construct a simple and efficient MRIP protocol for $\OSAT$.  As in classical MIP protocols, an MRIP protocol for any language $L\in
\sf{NEXP}$ can be obtained by first reducing $L$ to $\OSAT$ and then using our
protocol.  
As our protocol is highly efficiently, the complexity of the overall protocol
for $L$ is the same as the reduction.  Our protocol for $\OSAT$ is defined in \figref{scoring}, and we have the following lemma.

\begin{figure}[htbp]

\centering

\fbox{ \begin{minipage}{0.96\textwidth}

{\normalsize \vspace{0.0ex} \noindent{}For any instance $B$, the protocol $(V,
\vec{P})$ works as follows: 
\begin{enumerate}[leftmargin=15pt]
\item \label{step-1} $P_1$ sends
$c\in \{0, 1\}$ and $a \in \{0, 1, \ldots,2^{r+3s}\}$ to $V$.  $V$ outputs
$c$ at the end of the protocol.

\item If $c=1$ and $a< 2^{r+3s}$, or if $c=0$ and $a=2^{r+3s}$, the protocol
ends, and $R = -1$.

\item \label{prover1assignment}
Otherwise, $V$ uniformly and randomly chooses two binary strings of length $r + 3s$, $w = (z,b_1,b_2,b_3)$ and $w' = (z',b_4,b_5,b_6)$, as well as a number
$k \in \{1,2, \dots, 6\}$.

$V$ sends
$b_1,b_2,b_3,b_4,b_5,b_6$ to $P_1$ and $b_k$ to $P_2$.

\item \label{step-4}
$P_1$ sends to $V$ six bits, $A(b_i)$ with $i\in \{1, 2,\dots, 6\}$, and
$P_2$ sends one bit, $A'(b_k)$.

\item
The protocol ends and $V$ computes the payment $R$ as follows.
\begin{enumerate}
\item \label{step-5a}
If $A(b_k)\neq A'(b_k)$ then $R = -1$.

\item \label{step-5b}
Otherwise, if $B(z,b_1,b_2,b_3,A(b_1),A(b_2),A(b_3)) = 0$ then $R=0$.

\item \label{step-5c}
Else, let $b = B(z',b_4,b_5,b_6,A(b_4),A(b_5),A(b_6))$, $p_1 = a/2^{r+3s}$, and $p_0 = 1-p_1$.

$V$ computes $R$ using BSR. If
$b=1$, $R = \frac{2p_1 - (p_1^2 + p_0^2) + 1}{11}$, else $R = \frac{2p_0 - (p_1^2 + p_0^2) + 1}{11}$.
\vspace{1ex}
\end{enumerate}
\end{enumerate}

}
\end{minipage}
}

\caption{A simple and efficient MRIP protocol for $\OSAT$.}
\figlabel{scoring}
\end{figure}

\begin{lemma}
    \lemlabel{nexpmrip}
    $\OSAT$ has a 2-prover 3-round MRIP protocol where, for any instance $B$ of length~$n$, the randomness used by the verifier, the computation complexity,
and the communication complexity of the protocol are all $O(n)$.
Moreover, the evaluation of the payment function consists of constant number of arithmetic operations over $O(n)$-bit numbers.
\end{lemma}

\begin{proof}
For any instance $B$ with $r+3s+3$ variables (thus $n\geq r+3s+3$),
the provers can,
with their unbounded computation power, find an oracle $A^*$ that maximizes the number of satisfying $(r+3s)$-bit strings for $B$. Denote this number by $a^*$.
If $B\in \OSAT$ then $a^* = 2^{r+3s}$, otherwise $a^*< 2^{r+3s}$.

Roughly speaking, in our MRIP protocol in~\figref{scoring}, the verifier incentivizes
the provers to report the correct value of $a^*$, so that the membership of $B$ can be decided.
To see why this is the case,
let $s^*$ be one of the best strategy profiles of the provers.
Then $s^*$
must satisfy
\begin{equation}\label{equ:1}
\mbox{either } c=1 \mbox{ and } a=2^{r+3s}, \mbox{ or } c=0 \mbox{ and } a<2^{r+3s}.
\end{equation}
Otherwise, the provers' expected payment is $-1$. Meanwhile, by sending $c=0$ and $a=0$ in Step~\ref{step-1}
and all 0's in Step \ref{step-4}, their expected payment is 0.

Now we consider which of the two cases in Equation \ref{equ:1} the provers should report.
Note that $P_2$ only answers one query of the verifier (in Step \ref{step-4}). 
Thus under any strategy $\tilde{s}_2$ and given any $c$ and~$a$, $P_2$ de facto commits to an
oracle $A': \{0,1\}^{s} \rightarrow \{0,1\}$.  Assume that
$P_1$, using a strategy $\tilde{s}_1$ and seeing $(b_1,...,b_6)$,
sends $V$ six bits in Step \ref{step-4} that are not consistent with $A'$ ---that is,
there exists $i\in\{1,\ldots,6\}$ such that
$A(b_i)\neq A'(b_i)$.
Let $q$ be the probability that, conditioned on $(b_1,...,b_6)$, the verifier chooses a $k$ that catches the provers in Step \ref{step-5a};
we have $q\geq 1/6$.
Let $R$ be the payment to the provers conditioned on $(b_1,...,b_6)$ and on the event that they are not caught in Step \ref{step-5a}.
Note that $R \leq \frac{2}{11}$ by the definition of Brier's scoring rule.
Thus the expected payment to the provers conditioned on $(b_1,...,b_6)$ is $-q + (1-q)R<0$.
However, if $P_1$ answers the verifier's queries consistently with $A'$, their expected payment conditioned on $(b_1,...,b_6)$ is non-negative.
Accordingly, the best strategy profile $s^*$ must be such that, for any $c$, $a$ and the oracle committed by $P_2$, $P_1$'s answers for
any $(b_1,...,b_6)$ are always consistent with $A'$.
Thus, under $s^*$ the payment is never computed in
Step \ref{step-5a}.

Whether or not $B$ evaluates to 0 in Step \ref{step-5b} is determined solely
by $b_1$, $b_2$, $b_3$ and $A'$.
If $B$ evaluates to 0, then it does not matter what $a$ or $c$ is, and the provers' received payment is 0. 
If $B$ does not evaluate to 0 in Step \ref{step-5b}, then the expected payment to the provers in Step \ref{step-5c}
is defined by Brier's scoring rule: the true distribution of $b$, denoted by $D$, is such that
$D(1) = a'/2^{r+3s}$,
with $a'$ being the number of satisfying $(r+3s)$-bit strings for $B$ under oracle $A'$; the realized value is
$b = B(z',b_4,b_5,b_6,A(b_4),A(b_5),A(b_6))$; and the reported distribution is $(p_1, p_0)$.
Indeed, since $b_4, b_5, b_6$ are independent from $b_1, b_2, b_3$, we have that $w'$ is a uniformly random input to $B$, and the probability for $b$ to be 1 is exactly $a'/2^{r+3s}$.
Since Brier's scoring rule is
strictly proper, conditioned on $A'$, the provers maximize the expected payment by reporting
\begin{equation}\label{equ:2}
a=a',
\end{equation}
 which implies $(p_1, p_0) = (D(1), D(0))$.

If $B\not\in \OSAT$, then no matter which oracle $A'$ is committed under $s^*$, we have $a'< 2^{r+3s}$.
By Equations \ref{equ:1} and \ref{equ:2}, $a<2^{r+3s}$ and $c=0$
as desired.

If $B\in \OSAT$, which is the more interesting part, 
we show that under $s^*$ prover $P_2$ commits to the desired 3-satisfying oracle $A^*$ (so that $a' = 2^{r+3s}$ and $D(1)=1$).
Let $\sf{BSR}(D)$ denote the expected score for reporting $D$ under BSR, when $D$ is the true distribution.
\begin{eqnarray}\label{equ:BSR_3}
\sf{BSR}(D) &=& D(1)[2D(1)-D(1)^2-(1-D(1))^2-1]  \\
& & + (1-D(1))[2(1-D(1))-D(1)^2-(1-D(1))^2-1] \nonumber \\
 &=& 2(D(1)^2-D(1)). \nonumber
\end{eqnarray}
Thus $\sf{BSR}(D)$ is symmetric at $D(1)=1/2$, strictly decreasing on $D(1)\in [0, 1/2]$, strictly increasing on $D(1)\in [1/2, 1]$, and maximized when $D(1) = 1$ or $D(1) = 0$.
Note that the shifting and scaling of $\sf{BSR}$ in Step \ref{step-5c} do not change these properties, but make $\sf{BSR}(D)$ strictly positive when $D(1)=1$ or $D(1)=0$.
Therefore, to maximize their expected payment conditioned on the event that Step \ref{step-5c} is reached, $P_2$ should commit to either an oracle $A'$ such that $D(1)$ is as small as possible, or an $A'$ such that $D(1)$ is as large as possible, whichever makes $D(1)$ further from $1/2$.

If there is no oracle $A'$ such that $a'=0$, then the only way for the provers to maximize their expected payment is to commit to the 3-satisfying oracle $A^*$ (thus $a'=1$), under which Step \ref{step-5c} is reached with probability 1. Again by Equations \ref{equ:1} and \ref{equ:2}, we have $c=1$ and $a = 2^{r+3s}$.

If there are both a 3-satisfying oracle $A^*$ and an oracle $A'$ such that $a'=0$, we need to make sure that $P_2$ does not commit to $A'$.
To do so, we use $w$ along with Step \ref{step-5b}.
In particular, committing to any oracle other than $A^*$ or $A'$ 
results in an expected payment strictly smaller than that by committing to $A^*$,
since it increases the probability that the protocol ends at Step \ref{step-5b} with $R=0$,
and strictly decreases the expected payment conditioned on Step \ref{step-5c} being reached.
Moreover, if $P_2$ commits to $A'$, then $B$ {\em always} evaluates to 0 in Step \ref{step-5b},
 and Step \ref{step-5c} is actually never reached.
 Thus, even though by committing to $A'$ the provers maximize their expected payment in Step \ref{step-5c}, their actual expected payment is 0.
Instead, by committing to $A^*$, Step \ref{step-5c} is reached with probability 1 and the provers get positive payment.
Accordingly, the strategy profile $s^*$ must be such that
$P_2$ commits to $A^*$ and $P_1$ sends $a=2^{r+3s}$ and $c=1$,
as desired.
If there are multiple 3-satisfying oracles for $B$,
then the provers can pre-agree on any one of them (by taking the first in lexicographical order, for example).

In sum,
$(V, \vec{P})$ is an MRIP protocol for $\OSAT$.
Since $n\geq r+3s+3$, the number of coins flipped by $V$ for sampling $w$, $w'$, and $k$ is $O(n)$,
and so is the number of bits exchanged between $V$ and $\vec{P}$.
Moreover, given an input string $w = (z,b_1,b_2,b_3)$ for $B$ and the 3-bit answers of the oracle for $b_1, b_2, b_3$,
$B$ can be evaluated in linear time. Thus the running time of $V$ is $O(n)$ plus a constant number of arithmetic operations to compute the payment in Step \ref{step-5c}.
\end{proof}

\pparagraph{Remarks}
There is a tradeoff between the utility gap and the computational efficiency in the two MRIP protocols we have constructed for $\sf{NEXP}$.
The protocol in \figref{simple-nexp} has constant utility gap but relies on the MIP protocol, which has high (even though polynomial) communication and computation overheads beyond the reduction to $\OSAT$.
On the other hand, the protocol in \figref{scoring} is very efficient,
with just linear computation and communication overheads beyond the reduction to $\OSAT$,
but has exponential utility gap.
It would be interesting to see if there exists an MRIP protocol for $\sf{NEXP}$ that has constant or noticeable utility gap and is highly efficient
 (e.g., with linear overhead beyond the reduction to $\OSAT$).

To the best of our knowledge, the property of~$\sf{BSR}$ in Equation~\ref{equ:BSR_3} has never been discussed in the literature. 
All existing uses of proper scoring rules
are with respect to a fixed distribution and have the expert report the truth about that distribution.
In contrast, our use of scoring rules compares the expected scores across \emph{different} distributions: 
by committing to different oracles, the expert can choose which distribution is the true distribution,
and can tell the truth about that distribution to maximize his corresponding score. 
The correctness of our protocol depends on the expert 
committing to the distribution with the highest score under truth-telling.

\section{MRIP with Constant and Noticeable Utility Gap}\label{sec:const-polygap}

We have shown in~\secref{nexp} that the class of MRIP protocols with constant utility gaps contains both $\sf{NEXP}$ and $\sf{coNEXP}$, making them
more powerful than classic MIP protocols. In this section, we characterize the exact power of the classes
of MRIP protocols with constant and polynomial utility gaps.
That is,
we prove~\thmref{constantchar} and~\thmref{polychar}: 
\[
    \cMRIP = \sf{P^{||NEXP[O(1)]}} \quad \mbox{ and  } \quad \pMRIP = \sf{P^{||NEXP}}.
\]

%

To do so, let $\alpha(n)$ be a function of $n$, which (1) only takes positive integral values, 
(2) is upper-bounded by a polynomial in $n$, and (3) is polynomial-time computable.%
\footnote{To prove~\thmref{constantchar} and~\thmref{polychar},
 we only need $\alpha(n)$ to be constant or polynomial in $n$.
However, the lemmas in this section hold for all $\alpha(n)$'s that are polynomial-time computable (given $1^n$) and polynomially bounded. 
For example, $\alpha(n)$ can be $\log n$, $\sqrt{n}$, etc.}
We refer to the class of languages that have an MRIP protocol with $O(\alpha(n))$ utility gaps as $\aMRIP$,

Recall that $\sf{P^{||NEXP[\alpha(n)]}}$ is the class of languages decidable by polynomial-time Turing machines
making $O(\alpha(n))$ non-adaptive queries to an $\sf{NEXP}$ oracle. We prove tight upper- and lower-bounds
on the power of the class $\aMRIP$.

%

\begin{lemma}\label{lem:notc-lower}
$ \sf{P^{||NEXP[\alpha(n)]}} \subseteq \aMRIP$.
\end{lemma}

\begin{proof}
Consider any language $L \in \sf{P^{||NEXP[\alpha(n)]}}$.
Let $M$ be a polynomial-time Turing machine deciding $L$, with access to an oracle $O$ for an $\sf{NEXP}$ language.
Without loss of generality,
$M$ makes exactly $\alpha(n)\geq 1$ non-adaptive queries to $O$.
The MRIP protocol for $L$ uses our MRIP protocol for $\sf{NEXP}$ to simulate the oracle, as in~\figref{constantgap}.

%

\begin{figure}[phtb]

\centering
\fbox{
\begin{minipage}{0.96\textwidth}
{\normalsize
\vspace{0.5ex}
\noindent{}For any input string $x$ of length $n$, the protocol $(V, \vec{P})$ works as follows.
Initially $R_n = 0$.
\begin{enumerate}[leftmargin=15pt]
\item $P_1$ sends a bit $c\in\{0,1\}$ to $V$.  $V$ outputs $c$ at the end of the protocol.

\item $V$ simulates $M$ on $x$ till $M$ outputs $\alpha(n)$ queries for $O$, denoted by $q_1,\dots, q_{\alpha(n)}$.
\item To answer $M$'s oracle queries, for each
$i\in \{1, 2,\dots, \alpha(n)\}$, $V$ does the following:
\begin{enumerate}
\item $V$ first reduces $q_i$ to an $\OSAT$ instance $\phi_i$ (whose length is polynomial in $n$).
\item\label{step:nexp} $V$ sends $\phi_i$ to $P_1$ and $P_2$ and
executes the MRIP protocol for $\sf{NEXP}$ in~\figref{simple-nexp}.
Let $c^*_i$ and $R^*_i$ be the answer bit and the payment in that protocol respectively.
$V$ returns $c^*_i$ as the oracle's answer for $q_i$,
and
updates the sum $R_n \leftarrow R_n + R^*_i$.

\end{enumerate}
\item\label{step:final}
$V$ continues simulating $M$ till the end.
If $c$ does not match $M$'s output, then the protocol ends with reward $R =-1$;
otherwise the protocol ends with $R = R_n/\alpha(n)$.
\end{enumerate}
}
\end{minipage}
}
\caption{An MRIP protocol for $\sf{P^{||NEXP[\alpha(n)]}}$.}
\figlabel{constantgap}
\end{figure}

To see why this protocol works, first note that reporting the correct answer bit $c$ and answering all $\alpha(n)$ $\sf{NEXP}$
queries $q_1,\dots, q_{\alpha(n)}$ correctly leads to a reward $R \ge 1/2$ for the provers. In particular,
according to our protocol in~\figref{simple-nexp} and the proof of Lemma \ref{lem:nexp-mip},
if the provers use the optimal strategy 
for each query $q_i$ (which includes sending the correct answer bit $c^*_i$),
the provers get $R^*_i = 1$ if $\phi_i \in \OSAT$
and $R^*_i = 1/2$ if $\phi_i \notin \OSAT$.

Now, suppose the
provers report an incorrect answer bit $c' \neq c$ at the beginning.
Then, either (a) the output of $M$ in Step~\ref{step:final} does not match $c'$, and thus $R=-1$; or
(b) there exists an $\sf{NEXP}$ query $q_i$ such that the answer bit $c^*_i$
in Step~\ref{step:nexp} is incorrect.

In case (a), the provers' expected payment loss is at least $1/2+1 = 3/2>1/\alpha(n)$,
as $\alpha(n) \geq 1$. In case (b),
because the protocol in~\figref{simple-nexp} has $O(1)$ utility gap,
the provers' expected payment loss in the overall protocol is at least
$1/O(\alpha(n))$.
Thus, the provers' optimal strategy is to report the correct answer bit $c$ and to answer all $\alpha(n)$ $\sf{NEXP}$ queries correctly.
\end{proof}

To complete the characterization, we prove a tight upper-bound for $\aMRIP$.

\begin{lemma}\label{lem:notc-upper}
$ \aMRIP \subseteq \sf{P^{||NEXP[\alpha(n)]}}$.
\end{lemma}

\begin{proof}
Given any $L \in \aMRIP$, let $(V, \vec{P})$ be
the MRIP protocol with $O(\alpha(n))$ utility gap for~$L$. 
Again without loss of generality, assume the utility gap is exactly $\alpha(n)$.
To prove Lemma \ref{lem:notc-upper}, we simulate $(V, \vec{P})$ using a $\sf{P^{||NEXP[\alpha(n)]}}$ Turing machine.

Consider the following deterministic oracle Turing machine $M$.
Given any input $x$ of length $n$, $M$ divides $[-1,1]$ into $4 \alpha(n)$ intervals, each of length $1/(2\alpha(n))$.
That is, the $i$th interval is $[i/2\alpha(n),(i+1)/2\alpha(n))$ for each $i \in \{-2\alpha(n), \ldots, 2\alpha(n)-1\}$.\footnote{To include $1$ as a possible reward, interval $2\alpha(n)-1$ should be closed on both sides; we ignore this for simplicity.}
For each interval $[i/2\alpha(n),(i+1)/2\alpha(n))$, referred to as \emph{interval} $i$, $M$ makes the following queries to an $\sf{NEXP}$ oracle:
\smallskip
\begin{enumerate}[noitemsep,nolistsep,leftmargin=*]
\item\label{step:exists} { Does there exist a strategy profile $\tilde{s}$ in $(V, \vec{P})$ with expected payment $u(\tilde{s}; x)$ in interval $i$?}

\item\label{step:ans} { Does there exist a strategy profile $\tilde{s}$ in $(V, \vec{P})$ with expected payment $u(\tilde{s}; x)$ in interval $i$
and corresponding answer bit $c=1$?}
\end{enumerate}
\smallskip

Note that $M$ makes $O(\alpha(n))$ non-adaptive queries, each of polynomial size: indeed,
$M$ only needs to specify $x$, the value $i$ and the query index.
Some of these queries may turn out to be unnecessary in the end, but they are made anyway so as to preserve non-adaptivity.
%

We now show that the queries made by $M$ can be answered by an $\sf{NEXP}$ oracle.
Recall that in an MRIP protocol,
a strategy $\tilde{s}_{jk}$ of each prover $P_j$ for each round $k$ is a function mapping the transcript $P_j$ has seen at the beginning of round $k$
to the message he sends in that round.
Since the protocol has polynomially many provers and polynomially many rounds,
a strategy profile $\tilde{s}$ consists of polynomially many functions from $\{0,1\}^*$ to $\{0,1\}^*$,
and for each function, both the input length and the output length are polynomial in $n$.
Thus it takes at most exponentially many bits to specify a strategy profile: if the input length is at most $p(n)$ and the output length is at most $q(n)$, then $2^{p(n)}q(n)$ bits are sufficient to specify the truth table of a function.

Thus, an $\sf{NEXP}$ machine can
non-deterministically choose a strategy profile $\tilde{s}$.
It then goes through all possible realizations of $V$'s random string and, for each realization,
simulates $(V, \vec{P})$ on input~$x$ using~$\tilde{s}$, to compute the reward $R$.
Finally, the $\sf{NEXP}$ machine computes the expected payment $u(\tilde{s}; x)$, checks if $u(\tilde{s}; x)$ is in
 interval $i$ (and if $c=1$ for query~\ref{step:ans}),
 and accepts or rejects accordingly.
It is easy to see that if the desired strategy profile $\tilde{s}$ exists then this machine accepts $\tilde{s}$; otherwise it always rejects.

Since the verifier's random string has polynomially many bits, there
are exponentially many realizations in total.
Since $V$ runs in
polynomial time and it takes exponential time to look up the truth tables for $\tilde{s}$,
each realization takes exponential time to run.
Thus this machine runs in non-deterministic exponential time, and $M$'s queries can be answered by an $\sf{NEXP}$ oracle.



Finally, given the oracle's answers to its queries,
$M$ finds the highest index $i^*$ such that interval $i^*$ is ``non-empty'': that is, the oracle has answered 1 for query~\ref{step:exists} for this interval.
$M$ accepts if the oracle's answer to query~\ref{step:ans} for this interval is $1$, and rejects otherwise.
It is clear that $M$ runs in polynomial time.

The only thing left to show is that $M$ decides $L$ given correct answers to its oracle queries.
By definition, for the best strategy profile $s^*$ of the provers in $(V, \vec{P})$ for $x$, $u(s^*;x)$ falls into interval $i^*$.
Because $(V, \vec{P})$ has $\alpha(n)$ utility gap and each interval is of length $1/(2\alpha(n))$, by Definition~\ref{def:rewardgap},
all strategy profiles whose expected payments are in interval $i^*$ must have the same answer bit $c$ as that in $s^*$.
By the definition of MRIP protocols, $x\in L$ if and only if $c=1$, which occurs
if and only if the oracle's answer to query~\ref{step:ans} for interval $i^*$ is 1.
Thus $M$ decides $L$ and Lemma \ref{lem:notc-upper} holds.
\end{proof}

\begin{proof}[Proofs of \thmref{constantchar} and \thmref{polychar}]
\lemref{notc-lower} and \lemref{notc-upper} together imply that,
for any positive integral function $\alpha(n)$ that is polynomially bounded
and polynomial-time computable,
\[
    \aMRIP =\sf{P^{||NEXP[\alpha(n)]}}.
\]
\thmref{constantchar} holds by taking $\alpha(n) = O(1)$; and \thmref{polychar}
holds because $\pMRIP = \bigcup_{\alpha(n)=n^k: k\geq 0} \aMRIP = \bigcup_{\alpha(n)=n^k: k\geq 0} \sf{P^{||NEXP[\alpha(n)]}} = \sf{P^{||NEXP}}$.
\end{proof}

%
%

\section{Full Power of Multi-Prover Rational Interactive Proofs}\seclabel{mripchar}

In this section we prove \thmref{expchar}, that is, $\sf{MRIP = EXP^{||NP}}$.
We first show that $\sf{MRIP}$ is the same as another complexity class,  $\sf{EXP^{||poly-NEXP}}$,
which we define below.
We complete the proof of \thmref{expchar} by showing $\sf{EXP^{||NP}} = \sf{EXP^{||poly-NEXP}}$.
\begin{definition}\deflabel{exppolynexp}
    $\sf{EXP^{||poly-NEXP}}$ is the class of languages decidable by an exponential-time Turing machine with non-adaptive access to an $\sf{NEXP}$ oracle, such that
the length of each oracle query is polynomial in the length of the input of the Turing machine.
\end{definition}

\subsection{Preliminaries for Our Lower Bound}
In the lemma below, we first provide a lower bound on the class $\sf{MRIP}$.  In Section~\ref{sec:mripupper} we give a matching upper bound,
leading to a tight characterization.

\begin{lemma}
    \lemlabel{mriplow}
$\sf{EXP^{||poly-NEXP}}\subseteq\sf{MRIP}$.
\end{lemma}

To prove \lemref{mriplow}, let us recall some definitions and results from the literature of circuit complexity. 
First of all, a {\em circuit family} $\{C_n\}_{n=1}^{\infty}$ is a sequence of Boolean circuits
such that $C_n: \{0,1\}^n \rightarrow \{0,1\}$.
The gates are of types AND, OR, and NOT, with fan-ins 2, 2, and 1 respectively.
The input to a circuit is connected to a special set of ``input gates'', one for each bit of the input, whose output value is always the value of the corresponding bit.
The {\em size} of a circuit $C$ is the number of gates in $C$, including the input gates.
We index the gates in a circuit of size $g$ using
$\{1,2,...,g\}$.
Without loss of generality we assume that gate $g$ is the output gate of the whole circuit.
Moreover, if $C$ has input length $n$, without loss of generality we assume that gates $1,2,..,n$ are the input gates.
Note that the number of wires in $C$ is at most $2g$, since each gate has fan-in at most 2.  Thus we index the circuit's wires using $\{1,2,...,2g\}$.

\begin{definition}[DC uniform circuits \cite{arora2009computational}]
    \deflabel{polycircuits}
A circuit family $\{C_n\}_{n=1}^\infty$ is a {\em Direct Connect uniform (DC uniform) family} if the following questions can be answered in time polynomial in $n$:
\begin{enumerate}[noitemsep,nolistsep,leftmargin=*]
    \item \emph{SIZE}$(n)$: what is the size of $C_n$?
    \item \emph{INPUT}$(n, h, i)$: is wire $h$ an input to gate $i$ in $C_n$?
    \item \emph{OUTPUT}$(n, h, i)$: is wire $h$ the output of gate $i$ in $C_n$?
    \item \emph{TYPE}$(n, i, t)$: is $t$ the type of gate $i$ in $C_n$?
\end{enumerate}
\end{definition}

That is, the circuits in a DC uniform family may have exponential size,
but they have a succinct representation such that a polynomial-time Turing machine can answer all the questions in \defref{polycircuits}.
The class $\sf{EXP}$ can be characterized by the class of DC uniform circuit families:

\begin{lemma}[\hspace{.1pt}\cite{arora2009computational}]
    \label{lem:DC}
For any language $L$, $L\in \sf{EXP}$ if and only if it can be computed by
a DC uniform circuit family of size $2^{n^{O(1)}}$.
\end{lemma}

Next, we prove the following lemma, which is used in the proof of \lemref{mriplow}.

\begin{lemma}\lemlabel{expmrip25}
Every language $L$ in $\sf{EXP}$
has an MRIP protocol with two provers and five rounds based on DC uniform circuit families. 
\end{lemma}

\begin{proof}
By~\lemref{DC}, there exists a DC uniform circuit family $\{C_n\}_{n=1}^\infty$ that computes $L$. Let $g = 2^{n^k}$ be the size of each $C_n$, where $k$ is a constant that may depend on $L$.
We call a gate $i'\in \{1,2,...,g\}$ of $C_n$ an {\em input gate of gate $i$} if there is a directed wire from $i'$ to $i$.
For any input string $x$ of length $n$ and any gate $i$ in $C_n$,
let $v_i(x)\in \{0, 1\}$ be the value of $i$'s output on input $x$. In particular, $v_i(x) = x_i$ for any $i\in \{1,2,...,n\}$.
The 2-prover 5-round MRIP protocol $(V, \vec{P})$ for $L$ is given in \figref{fig:expmrip}. 

\begin{figure}[tbhp]

\centering
\fbox{
\begin{minipage}{0.96\textwidth}
{\normalsize

\noindent
For any input string $x$ of length $n$,
\begin{enumerate}[leftmargin=15pt]

\item $P_1$ sends one bit $c \in \{0,1\}$ to $V$. $V$ outputs $c$ at the end of the protocol.

\item\label{step:pick_gate}
$V$ computes $g = $SIZE$(n)$, picks a
    gate $i\in \{1,2,...,g\}$ uniformly at random, and sends $i$ to $P_1$.
    That is, $V$ queries $P_1$ for:
\begin{enumerate}
\item\label{step-2a}
 the type of gate $i$,

\item\label{step-2b}
the input gates and input wires of $i$, and

\item\label{step-2c}
the values of gate $i$ and its input gates.
\end{enumerate}
\item\label{step-3}
$P_1$ sends to $V$: type $t_i\in \{\mbox{AND, OR, NOT, INPUT}\}$; gates $i_1, i_2\in \{1,2,...,g\}$;
wires $h_1, h_2\in \{1,2,...,2g\}$; and values $v_i(x), v_{i_1}(x), v_{i_2}(x)\in \{0, 1\}$.

\item \label{step_p2}
$V$ picks a gate $i'\in \{i, i_1, i_2\}$ uniformly at random and sends $i'$ to $P_2$.

\item\label{step-7}
$P_2$ sends $v_{i'}'(x)\in \{0, 1\}$ to $V$.

\item
The protocol ends and $V$ computes the payment $R$ by verifying the following statements:
\begin{enumerate}
\item\label{step-6a} $t_i$ is the correct type of $i$
and the set of input gates of $i$ is correct using DC uniformity;
\item\label{step-6c}
if $i\in \{1,2,...,n\}$ (that is, an input gate of the circuit), then $v_{i}(x)= x_{i}$;

\item\label{step-6b}
if $i=g$ (that is, the output gate of the circuit), then $v_i(x)= c$;

\item\label{step-6d}
if $t_i \in \{\mbox{AND, OR, NOT}\}$, $v_i(x)$ follows the correct logic based on $t_i$ and $i$'s inputs.

\item\label{step-6e} 
The answers of $P_1$ and $P_2$ on the value of gate $i'$ are consistent.
\end{enumerate}
If any of these verifications fails then $R=0$; otherwise $R = 1$.
\end{enumerate}

}
\end{minipage}
}
\caption{An MRIP protocol for $\sf{EXP}$.}
\figlabel{fig:expmrip}
\end{figure}

To see why it is an MRIP protocol, notice that if $P_1$ and $P_2$ send the correct $c$ and always
answer $V$'s queries correctly according to $C_n$, then the payment to them is always $R=1$, irrespective of $V$'s coin flips.
Thus the expected payment is~$1$.
Below we show that any other strategy profile makes the expected payment strictly less than 1.

First of all, when the gate $i$ chosen by the verifier in Step \ref{step:pick_gate} is not an input gate,
if any of $P_1$'s answers in Step \ref{step-3} to queries \ref{step-2a} and \ref{step-2b} (namely, about $i$'s type, input gates and input wires) is incorrect,
then by DC uniformity the verification in Step \ref{step-6a} will fail, giving the provers a payment $R=0$.
Indeed, to verify whether $i_1$ and $i_2$ are the input gates of $i$,
it suffices to verify whether $h_1$ and $h_2$ are both the input wires of~$i$ and the output wires of $i_1$ and $i_2$:
this is why $V$ queries $P_1$ about $i$'s input wires.
Accordingly, if such a gate $i$ exists then the expected payment to the provers will be at most
$1-1/g < 1$.

Similarly, if there exists a non-input gate $i$ such that $P_1$ answers queries \ref{step-2a}
and \ref{step-2b} correctly but the values $v_i(x), v_{i_1}(x), v_{i_2}(x)$ are
inconsistent with $i$'s type, then Step \ref{step-6d} will fail conditioned on gate~$i$ being chosen,
and the expected payment to the provers is at most $1-1/g<1$.
Moreover, if there exists an input gate $i$ such that $v_i(x)\neq x_i$, or if $v_g(x)\neq c$,
then conditioned on gate $i$ being chosen, the expected payment is again at most $1-1/g<1$.

Next, as in the proof of \lemref{nexpmrip}, $P_2$ is only queried once (in Step \ref{step-7}).
Thus $P_2$ de facto commits to an oracle $A: \{1,\ldots,g\} \rightarrow \{0, 1\}$, which maps each gate to its value under input $x$.
If there exists a gate $i$ such that the values $v_i(x), v_{i_1}(x), v_{i_2}(x)$ in Step \ref{step-3} are not consistent with $A$,
then, conditioned on $i$ being chosen in Step \ref{step:pick_gate},
Step \ref{step-6e} will fail with probability $1/3$.
Since $i$ is chosen with probability $1/g$, the expected payment will be at most
$1-\frac{1}{3g}<1$.

Thus, the only strategy profile $\tilde{s}$ that can have expected payment equal to $1$ is the following:
\begin{enumerate}[noitemsep, nolistsep, leftmargin=*]
\item
 $P_1$ and $P_2$ report values of gates using the same oracle $A: \{1,\ldots,g\} \rightarrow \{0, 1\}$,
\item
 $A(i) = x_i$ for any input gate $i$,
\item
$A(g) = c$ for the output gate, and
\item
for any other gate $i$, $A(i)$ is computed correctly based on $i$'s type and input gates in $C_n$.
\end{enumerate}
Thus, $A(g)$ is computed according to $C_n$ with input $x$, and $A(g)=1$ if and only if $x\in L$.
Since $c=A(g)$, we have that $c=1$ if and only if $x\in L$ and $(V, \vec{P})$ is an MRIP protocol for $L$.
\end{proof}



\subsection{Lower Bound for $\sf{MRIP}$}
Using the protocol in \figref{fig:expmrip} as a building block, we are now ready to prove~\lemref{mriplow}.


\pparagraph{Circuits for $\sf{EXP^{||poly-NEXP}}$}
We start by creating some circuit structures for the class $\sf{EXP^{||poly-NEXP}}$.
For any language $L\in \sf{EXP^{||poly-NEXP}}$,
let $M$ be an exponential-time oracle Turing machine that decides $L$ using an oracle $O$.
Without loss of generality, assume $O$ is for $\OSAT$.
Let $q(n)$ be the number of oracle queries made by $M$ on any input $x$ of length $n$,
and $p(n)$ be the length of each query.
By the definition of $\sf{EXP^{||poly-NEXP}}$, $q(n)$ can be exponential in $n$, while $p(n)$ is polynomial.
Without loss of generality, $p(n)\geq 5$. Let $\ell(n) = p(n)q(n)$.
When $n$ is clear from context, we refer to $\ell(n)$, $p(n)$ and $q(n)$ as $\ell$, $p$ and~$q$ respectively.

Since the oracle queries are non-adaptive, there exists an exponential-time-computable function $f:\{0, 1\}^*\rightarrow \{0, 1\}^*$
such that, for any $x\in \{0, 1\}^n$, $f(x)\in \{0, 1\}^{\ell}$ and $f(x)$ is the vector of oracle queries made by $M$ given $x$.
($f$ is exponential-time computable because we can run $M$ on $x$ until it outputs all the queries.)
As in Lemma~\ref{lem:DC}, there exists a DC uniform circuit family $\{C_n\}_{n=0}^\infty$ of size $2^{n^{O(1)}}$ that computes $f$, where for any $n$,
$C_n$ has $n$-bit input and $\ell$-bit output.
Without loss of generality, the gates of $C_n$ can be partitioned into $q$ sets, one for each oracle query, such that the output of a gate only affects the value of the corresponding query. This can be done by duplicating each gate at most an exponential number of times. The resulting circuit family is still DC uniform.
Also without loss of generality, the oracle queries are all different. This can be done by including the index $i\in \{1,\dots, q\}$ in the $i$th query.

Given the vector of oracle answers corresponding to the $q$ queries of $M$, $b\in \{0, 1\}^q$, the membership of $x$ can be decided in time exponential in $n$.
Let $f': \{0, 1\}^*\rightarrow \{0, 1\}$ be a function such that,
given any $(n+q)$-bit input $(x, b)$ where $|x|=n$ and $b$ is the vector of oracle answers $M$ gets
with input $x$, $f'(x, b)$ is the output of $M$. Again,
$f'$ is computable by a DC-uniform circuit family $\{C'_n\}_{n=1}^\infty$ of size $2^{n^{O(1)}}$,
where each $C'_n$ has $(n+q)$-bit input and 1-bit output.
The size of $C'_n$ is exponential in $n$ but may not be exponential in its own input length,
since $q$ may be exponential in~$n$.
In particular, the Turing machine that answers questions SIZE, INPUT, OUTPUT, TYPE for $C'_n$ runs in time polynomial in $n$ rather than $n+q$.

Given the two circuit families defined above, the membership of $x$ in $L$ can be computed by the following three-level ``circuit:'' besides the usual AND, OR, NOT gates,
it has $q$ ``$\sf{NEXP}$'' gates, each of which has a $p$-bit input and 1-bit output, simulating the $\OSAT$ oracle.

\smallskip
\begin{itemize}[noitemsep,nolistsep,leftmargin=*]
    \item Level 1: The circuit $C_n$ for computing $f$. We denote its output by $(\phi_1, \phi_2,..., \phi_{q})$, where each $\phi_i$ is of $p$ bits and is an instance of $\OSAT$. 
Let $g = 2^{n^k}$ be the size of $C_n$, where $k$ is a constant.
Similar to our naming convention before, the set of gates is $\{1, 2, ..., g\}$,
the set of input gates is $\{1, 2, ..., n\}$, and the set of output gates is $\{n+1, n+2, ..., n+\ell\}$.
The input and the output gates correspond to $x$ and $(\phi_1, \phi_2,..., \phi_{q})$ in the natural order.

\item
Level 2: We have $q$ $\sf{NEXP}$ gates, without loss of generality denoted by $g+1, g+2, ..., g+q$.
For each $i\in \{1,2,...,q\}$, gate $g+i$ takes input $\phi_i$ and outputs 1 if and only if $\phi_i\in \OSAT$.

\item
Level 3: The circuit $C'_n$ for computing $f'$.
Let $g' = 2^{n^{k'}}$ be the size of $C'_n$, where $k'$ is a constant.
The set of gates is $\{g+q+1, g+q+2, ..., g+q+g'\}$, the set of
input gates is $\{g+q+1,..., g+q+n, g+q+n+1,..., g+q+n+q\}$, and the output gate is gate $g+q+g'$.
The first $n$ input gates connect to $x$, and
the remaining ones connect to the $\sf{NEXP}$ gates of Level 2.
The output of $C'_n$ is the final output of the whole circuit.
\end{itemize}
\smallskip

Inside the three-level circuit, we can compute each output gate of Level 1 and Level 3 using the protocol in \figref{fig:expmrip},
and each $\sf{NEXP}$ gate in Level 2 using the protocol in \figref{simple-nexp}.
However, we need to show that there exists an MRIP protocol $(V, \vec{P})$
where the verifier can get a consistent answer
to {\em all} of them simultaneously. In particular, the provers should not lie in $C_n$
in order to change the input to the $\sf{NEXP}$ queries to gain a higher overall expected payment.

\pparagraph{Our protocol}
Our protocol is specified in \figref{expnexp}.  It uses four provers.
In this protocol the verifier needs to compute $q(n)$ and $p(n)$.
Without loss of generality, we assume $q(n) = 2^{n^d}$ for some constant $d$, so its binary representation can be computed in time polynomial in $n$.
Since $p(n)$ is a polynomial in $n$, it can be computed by a polynomial-time verifier.

\begin{figure}[tbhp]

\centering
\fbox{
\begin{minipage}{0.96\textwidth}
{\normalsize

\noindent
For any input string $x$ of length $n$,
\begin{enumerate}[leftmargin=15pt]

\item $P_1$ sends one bit $c \in \{0,1\}$ to $V$. $V$ outputs $c$ at the end of the protocol.

\item\label{char-step:pick_gate}
$V$ computes $g =$ SIZE$(C_n)$, $q(n)$, and $g' =$ SIZE$(C'_n)$.\\
$V$ picks a gate $i\in \{1,2,..., g+q+g'\}$ uniformly at random
and sends $i$ to $P_1$. \\
By doing so, $V$ queries $P_1$ for:
\begin{enumerate}
\item\label{char-step-2a}
 the type $t_i$ of gate $i$,

\item\label{char-step-2b}
the input gates and input wires of $i$, and 

\item\label{char-step-2c}
the values of gate $i$ and its input gates. 
\end{enumerate}
\item\label{char-step-3}
$P_1$ sends to $V$ the following:
\begin{enumerate}
\item type $t_i\in \{\mbox{AND, OR, NOT}, \mbox{INPUT}, \sf{NEXP}\}$;

\item input gates $i_1,i_2,\dots, i_{f(i)}$ and input wires $h_1, h_2,\dots, h_{f(i)}$,
where $f(i)$ is the number of input gates of type $t_i$; and

\item values of gate $i$ and its input gates: $v_i(x)$, $v_{i_1}(x), v_{i_2}(x), \ldots, v_{i_{f(i)}}(x)$.
\end{enumerate}

\item\label{char-step-4}
$V$ verifies the following using DC uniformity or the naming convention:
\begin{enumerate}
\item\label{char-step-4a}
$t_i$ is the correct type of $i$ (in particular, if $i\in \{g+1,...,g+q\}$ then $t_i= \sf{NEXP}$) and $f(i)$ is correct for $t_i$; and
\item\label{char-step-4b}
the set of input gates of $i$ is correct.
\end{enumerate}
If any of the verifications fails, the protocol ends and
$R=-1$.

\item \label{char-step_p2}
$V$ picks a gate $i'$ uniformly at random from $\{i\} \cup \{i_1\ldots, i_{f(i)}\}$, and sends $i'$ to $P_2$. 

\item\label{char-step-7}
$P_2$ sends $v_{i'}'(x)\in \{0, 1\}$ to $V$.

\item\label{char-step-7b} {\bf Consistency.}
$V$ verifies $v_{i'}(x)=v_{i'}'(x)$: that is, the answers of $P_1$ and $P_2$ on the value of gate~$i'$ are consistent. If not, the protocol ends and $R=-1$.

\item\label{char-step-8} {\bf Correctness (Non-$\sf{NEXP}$ gates).}
If $t_i\neq \sf{NEXP}$, then $V$ checks if $v_i(x)$ is computed correctly
from $v_{i_1}(x), v_{i_2}(x), \ldots, v_{i_{f(i)}}(x)$ as follows:

\begin{enumerate}
\item\label{char-step-8a}
if $t_i=\mbox{INPUT}$ then $v_i(x) = v_{i_1}(x)$,
and if $i$ is one of the first $n$ gates in $C_n$ or $C'_n$, then $v_i(x)$ equals the corresponding bit of $x$;

\item\label{char-step-8b}
if $t_i \in \{ \mbox{AND, OR, NOT}\}$, then $v_i(x)$ follows the logic between $i$ and its inputs.

\item\label{char-step-8c}
if $i=g+q+g'$ (i.e., the output gate of the whole circuit), then $v_i(x) = c$.
\end{enumerate}

The protocol ends with the following reward: if any of the verifications fails then
$R=-\frac{1}{p+1}$,
otherwise $R=\frac{1}{p+1}$,
where $p$ is the length of each $\sf{NEXP}$ query.

\item\label{char-step-9}{\bf Correctness ($\sf{NEXP}$ gates).}
If $t_i= \sf{NEXP}$, then $V$ first
checks if $\phi_i = (v_{i_1}(x), ..., v_{i_p}(x))$ forms a valid $\OSAT$ instance.%
\footnote{Without loss of generality, we assume that
the instances of $\OSAT$ have a canonical form.}
If not, the protocol ends with $R=-\frac{2}{p+1}$.\\
If $\phi_i$ is a valid $\OSAT$ instance, then $V$ sends $\phi_i$ to $P_3$ and
$P_4$ and runs the MRIP protocol for $\sf{NEXP}$ in \figref{simple-nexp}.  Let
$c^*$ and $R^*$ respectively be the output and the reward of the $\sf{NEXP}$
protocol.  If $c^*=v_i(x)$ then $R = \frac{2R^*}{p+1}$; otherwise
$R=-\frac{2}{p+1}$.
\end{enumerate}

} \end{minipage} } \caption{An MRIP protocol for $\sf{EXP^{||poly-NEXP}}$.}
\figlabel{expnexp} \end{figure}

To prove the correctness of the protocol in \figref{expnexp},
first note that for any input string $x$, no matter which gate $i$ is chosen by $V$ in Step \ref{char-step:pick_gate},
if the provers always give correct answers according to the computation of $C_n$, the $\sf{NEXP}$ gates and $C'_n$,  the payment to them is
$R\geq\frac{1}{p+1}> 0.$
The first inequality is tight when either (a) $i$ is not an $\sf{NEXP}$ gate,
or (b) $i$ is an $\sf{NEXP}$ gate and the corresponding query $\phi_i$ is not in $\OSAT$ (since $R^*=1/2$ in this case).
If $i$ is an $\sf{NEXP}$ gate
and $\phi_i\in \OSAT$, then $R=\frac{2}{p+1}$ as $R^* = 1$.
Let $s$ be the strategy profile where the provers always send correct answers as
described above. Thus we have $u(s)\geq \frac{1}{p+1}$.



\pparagraph{The correctness of our protocol}
Arbitrarily fix a best strategy profile $s^*$ of the provers,
we show that under $s^*$, $c =1$ if and only if $x \in L$.

Since $P_2$ is queried only once (Step \ref{char-step-7}),
as in the proof of \lemref{nexpmrip},
any strategy of $P_2$ commits to an oracle $A: \{1,2,...,g+q+g'\}\rightarrow \{0,1\}$,
mapping each gate in the three-level circuit to its value under input~$x$.
First, we show that for non-$\sf{NEXP}$ gates,
$P_1$ answers all queries
consistently with~$A$. 

\begin{claim}\label{claim:consistent-non-nexp}
Under $s^*$, for any gate $i$ that is not an $\sf{NEXP}$ gate
and is
chosen by the verifier in Step~\ref{char-step:pick_gate},
$P_1$ reports the correct type and input gates of~$i$ in Step~\ref{char-step-3}, and reports the values of gate~$i$ and its input gates consistently with $A$.
\end{claim}

\begin{proof}
Suppose there exists a non-$\sf{NEXP}$ gate $i$ such that $P_1$ does not
report its type and input gates correctly.
Conditioned on $i$ being chosen by the verifier,
some verification in Step~\ref{char-step-4} is guaranteed to fail, and the payment is $-1$.
Consider the following alternative strategy $s_1'$ of $P_1$: if $i$ is not chosen by~$V$, then $P_1$'s strategy remains the same; if $i$ is chosen, then $P_1$ acts ``correctly'' as specified in Claim~\ref{claim:consistent-non-nexp}.
Under this strategy, when $i$ is chosen the payment is at least $-\frac{1}{p+1}>-1$, and when $i$ is not chosen the payment stays the same. Thus the expected payment gets larger, contradicting the fact that $s^*$ is the provers' best strategy profile.

Similarly, consider the case where $P_1$ reports $i$'s type and input gates correctly, but the reported values
do not match $A$ on some gate $i' \in \{i\} \cup \{i_1, \ldots, i_{f(i)}\}$.
Conditioned on gate $i$
being chosen, with probability at least $\frac{1}{f(i)+1} \ge \frac{1}{3}$,
$V$ picks $i'$  in Step~\ref{char-step_p2}
and the consistency check in Step~\ref{char-step-7b} fails,
leading to a payment of $-1$.
If $i'$ is not chosen in Step~\ref{char-step_p2},
the payment to the provers is at most $\frac{1}{p+1}$ (in Step \ref{char-step-8}).
Thus the expected payment conditioned on $i$ being chosen is at most
\[ - \frac{1}{3} + \frac{2}{3} \cdot \frac {1}{p+1} < -\frac {1}{p+1},\]
where the inequality holds since $p\geq 5$.
Again, consider the alternative strategy $s_1'$ of $P_1$. Under this strategy, conditioned on $i$ being chosen the expected payment is at least $-\frac{1}{p+1}$; and conditioned on $i$ not being chosen it stays the same. Thus the expected payment gets larger, again a contradiction.
%
\end{proof}

Below we only need to consider
cases where $P_1$ acts according to Claim \ref{claim:consistent-non-nexp}.
We argue about the correctness of $A$ on non-$\sf{NEXP}$ gates, and we have the following.
\smallskip

\begin{claim}\label{claim:non-nexp-a}
Under $s^*$, for every gate $i$ that is
not an $\sf{NEXP}$ gate, $A(i)$ and the values $A(i_1), \dots, A(i_{f(i)})$
are such that the verifications in Step \ref{char-step-8} succeed.
\end{claim}

\begin{proof}
By contradiction, assume this is not the case and
compare $s^*$ with the ``always correct'' strategy profile $s$ previously defined.
Recall that,
conditioned on $i$ being chosen,
for any non-$\sf{NEXP}$ gate~$i$
the payment under $s$ is exactly $\frac{1}{p+1}$,
and for any $\sf{NEXP}$ gate $i$
the payment under $s$ is at least $\frac{1}{p+1}$.

Under $s^*$, by Claim~\ref{claim:consistent-non-nexp},
$P_1$'s answers for $v_i(x), v_{i_1}(x),\dots, v_{i_{f(i)}}(x)$ are consistent with $A$.
If $A$ makes some verification in Step \ref{char-step-8} fail, then conditioned on $i$ being chosen,
the payment under $s^*$ is $-\frac{1}{p+1}$.
That is, the payment under $s^*$ drops by $\frac{2}{p+1}$ compared with that under $s$.

However, $s$ and $s^*$ may not have the same oracle queries to $\OSAT$. For each $\sf{NEXP}$ gate $j$
where the two queries differ, conditioned on $j$ being chosen,
the best case for $s^*$ (and the worst case for the analysis) is that its query $\phi^*_j$ is in $\OSAT$, resulting in payment $\frac{2}{p+1}$, while the query $\phi_j$ of $s$ is not in $\OSAT$, resulting in payment $\frac{1}{p+1}$.
That is, the payment under $s^*$ increases by $\frac{1}{p+1}$ compared with that under $s$.

Fortunately, for each $\sf{NEXP}$ gate $j$, in order for the two queries to differ,
there exists at least one non-$\sf{NEXP}$ gate $i$ in the part of the circuit $C_n$ for computing the input to $j$, where
the computation of $A$ (and thus $s^*$) is incorrect, and $A(i)$ and $A(i_1), \dots, A(i_{f(i)})$  make some verification in Step \ref{char-step-8} fail.
Otherwise the queries made by $A$ are computed correctly from the input $x$ and are the same as those under $s$.
Since gate $j$ and the corresponding gate $i$ are chosen with the same probability $\frac{1}{g+q+g'}$,
we have
\[
    u(s)-u(s^*) \geq \frac{1}{g+q+g'}\cdot \frac{2}{p+1} - \frac{1}{g+q+g'}\cdot \frac{1}{p+1} >0.
\]
If there is more than one such $j$, 
their corresponding gates $i$ are all different from each other, because the circuits for computing different oracle queries are disjoint from each other---so
the gap between $u(s)$ and $u(s^*)$ becomes even larger.
This contradicts that $s^*$ is the provers' best strategy, and thus Claim~\ref{claim:non-nexp-a} holds.
\end{proof}

%

\smallskip

Now we only need to consider
cases where $P_1$ acts according to Claims \ref{claim:consistent-non-nexp} and~\ref{claim:non-nexp-a}.
We prove the correctness of $A$ on $\sf{NEXP}$ gates. 

\smallskip

\begin{claim}\label{claim:nexp-a}
Under $s^*$, for every $\sf{NEXP}$ gate $i$,
$P_1$ reports the correct type and input gates of $i$ in Step~\ref{char-step-3},
and reports the values of gate~$i$ and its input gates consistently with $A$.
Moreover, $\phi_i =
(A(i_1), \dots, A(i_p))$ forms a valid $\OSAT$ instance and
$A(i)=1$ iff $\phi_i\in \OSAT$.
\end{claim}

\begin{proof}
The fact that $\phi_i$ forms a valid $\OSAT$ instance follows immediately from
Claims~\ref{claim:consistent-non-nexp} and~\ref{claim:non-nexp-a}, because
each bit of $\phi_i$ is the output of a logic gate and thus computed correctly from the input $x$ according to $C_n$.
We again compare $s^*$ with the always-correct strategy profile $s$.

Note that $A$ and $s$ are both correct on $C_n$, thus form the same $\OSAT$ queries.
They both evaluate $C_n'$ correctly as well, but
it is possible that $A$ has incorrect outputs of the $\sf{NEXP}$ gates and thus incorrect inputs to $C_n'$.
Nevertheless, for each non-$\sf{NEXP}$ gate $i'$, conditioned on $i'$ being chosen,
$s^*$ makes the verifications in Step~\ref{char-step-8} succeed, and
the payment is $\frac{1}{p+1}$ under both $s$ and $s^*$.

If $P_1$ reports $i$'s type and input gates incorrectly under $s^*$, then the payment is $-1$ (Step \ref{char-step-4}) conditioned on
$i$ being chosen. However, by reporting the required information correctly and reporting $v_i(x), v_{i_1}(x),\dots, v_{i_p}(x)$ consistently with $A$, the corresponding payment is at least $-\frac{2}{p+1}>-1$ and the expected payment increases, contradicting with the fact that $s^*$ is the provers' best strategy profile.

Suppose $P_1$ reports $i$'s type and input gates correctly,
but reports $v_{i'}(x)$ inconsistently with $A$ for some $i'\in \{i\}\cup\{i_1,\dots, i_p\}$.
In this case, with probability at least $\frac{1}{p+1}$ the payment is $-1$ (Step \ref{char-step-7b}),
and with probability at most $1-\frac{1}{p+1}$ the payment is at most $\frac{2}{p+1}$ (Step \ref{char-step-9}).
Thus the expected payment is
\[
    R \leq -\frac{1}{p+1} + (1-\frac{1}{p+1})\cdot \frac{2}{p+1} = \frac{1}{p+1} - \frac{2}{(p+1)^2}< \frac{1}{p+1}.
\]
The corresponding expected payment under $s$ is at least $\frac{1}{p+1}$.
As the two strategy profiles have the same payment $\frac{1}{p+1}$ conditioned on every non-$\sf{NEXP}$ gate
$i'$ being chosen,
we have $u(s)>u(s^*)$, a contradiction.

Finally, assume $P_1$ is consistent with $A$, but $A(i)$ is not the correct answer of $\phi_i$.
If the answer bit $c^*$ given by $P_3$ and $P_4$
is different from $A(i)$ (i.e., $v_i(x)$), then the payment is $-\frac{2}{p+1}<\frac{1}{p+1}$, less than the payment received under the always-correct strategy profile $s$.
If $c^*=v_i(x)$, then $c^*$ is the wrong answer bit in the MRIP protocol for $\sf{NEXP}$, and
the resulting payment $R^*$ is strictly less than the payment under $s$.
Thus, again we have that $u(s)>u(s^*)$, which is a contradiction, and Claim \ref{claim:nexp-a} holds.
\end{proof}

\smallskip

Claims~\ref{claim:consistent-non-nexp}, \ref{claim:non-nexp-a}, and~\ref{claim:nexp-a}
together imply that
the always-correct strategy profile $s$ is the only possibility for the provers' best strategy profiles; that is, $s^*=s$.
Under $s$, for any gate $i$, $A(i)$ is the correct value of $i$ under input $x$,
and $c = A(g+q+g')$. Thus $c=1$ if and only if $x\in L$, and~\lemref{mriplow} holds.

\pparagraph{Remark}
When proving~\thmref{mrip-23} in Section \ref{subsec:neg},
we show that any MRIP protocol can be simulated using only 2 provers.
In this section we still describe the protocol in \figref{expnexp} using 4 provers, to ease the analysis
and to avoid entangling the proofs of \thmref{expchar} and~\thmref{mrip-23}.

\subsection{Upper Bound for $\sf{MRIP}$}
\label{sec:mripupper}
We now give a tight upper-bound on $\sf{MRIP}$, leading to an exact characterization.

\begin{lemma}
\lemlabel{mripup} $\sf{MRIP}\subseteq\sf{EXP^{||poly-NEXP}}$.
\end{lemma}

\begin{proof}
The proof is similar to that of~\lemref{notc-upper}.
Let $L$ be a language with an MRIP protocol $(V,\vec{P})$.
Since $V$ runs in polynomial time, there exists a constant $k$ such that, for any two payments $R$ and $R'$
generated by $V$ on the same input of length $n$ and different random coins:
\[
R\neq R' \Rightarrow |R-R'|\geq \frac{1}{2^{n^k}}.
\]
For example, $n^k$ can be an upper bound on $V$'s running time.
Moreover, since $V$ uses polynomially many random coins, there exists a constant $k'$ such that
any payment that appears with positive probability under an input of length $n$ must appear with probability at least $\frac{1}{2^{n^{k'}}}$.
Thus, for an input $x$ of length $n$, and any two strategy profiles $s$ and $s'$,
where the expected payments $u(s; x)$ and $u(s'; x)$ are different,
\begin{equation}\label{equ:u}
|u(s; x) - u(s'; x)|\geq \frac{1}{2^{n^{k+k'}}}.
\end{equation}

Consider the following deterministic oracle Turing machine $M$:
given any input $x$ of length $n$, $M$ divides the interval $[-1,1]$
into $4\cdot 2^{n^{k+k'}}$ sub-intervals of length $\frac{1}{2\cdot 2^{n^{k+k'}}}$.
For any $i\in \{-2\cdot 2^{n^{k+k'}}+1, \ldots,2\cdot 2^{n^{k+k'}}\}$,
the $i$th interval is $\left[\frac{(i-1)}{2\cdot 2^{n^{k+k'}}}, \frac{i}{2\cdot 2^{n^{k+k'}}}\right]$.
For each interval $i$, $M$ makes the following two queries to an $\sf{NEXP}$ oracle:

\smallskip
\begin{enumerate}[noitemsep, nolistsep, leftmargin=*]
\item Does there exist a strategy profile $s$ in $(V, \vec{P})$ with expected payment $u(s; x)$ in interval $i$?
\item Does there exist a strategy profile $s$ in $(V, \vec{P})$ with expected payment $u(s; x)$ in interval $i$ and the corresponding answer bit $c=1$?
\end{enumerate}
\smallskip
$M$ makes exponentially many non-adaptive queries, and each
query has length polynomial in~$n$. Furthermore, each query can be answered by an $\sf{NEXP}$ oracle; see the proof of~\lemref{notc-upper}.

Given the oracle's answers, $M$ finds the highest index $i^*$ such that interval $i^*$ is non-empty: that is, the oracle's answer to the first query for interval $i^*$ is 1.
$M$ accepts if the answer to the second query for interval $i^*$ is $1$, and rejects otherwise. $M$ clearly runs in exponential time.

We show that $M$ decides $L$ given correct answers to its queries.
Similar to~\lemref{notc-upper}, by Definition~\ref{def:mrip}, the best strategy profile $s^*$ has the highest expected payment $u(s^*;x)$, which falls into interval~$i^*$.
By Inequality \ref{equ:u}, any strategy profile $s'$ with $u(s'; x)< u(s^*; x)$ has $u(s'; x)$ not in interval $i^*$, since the difference between
$u(s'; x)$ and $u(s^*; x)$ is larger than the length of the interval.
Thus, any strategy profile $s'$ with $u(s'; x)$ in interval $i^*$ satisfies
$u(s'; x) = u(s^*; x)$, i.e, they are all the best strategy profiles of the provers.
In particular, the answer bit $c$ is the same under all these strategy profiles,
and $c=1$ if and only if $x\in L$.
So the second query for interval $i^*$ is 1 if and only if $x\in L$, and $M$ decides~$L$.
\end{proof}

\subsection{Final Characterization}
So far we have established that $\sf{MRIP = EXP^{||poly-NEXP}}$.  To finish
the proof of \thmref{expchar}, we show $\sf{EXP^{||poly-NEXP}}$ equals $\sf{EXP^{||NP}}$.

\begin{lemma}\lemlabel{equaloracles}
$\sf{EXP^{||poly-NEXP}=EXP^{||NP}}$.
\end{lemma}

\begin{proof} 
First, we show $\sf{EXP^{||poly-NEXP}\subseteq EXP^{||NP}}$ using a padding argument.
Let $M_1$ be an exponential-time oracle Turing machine with non-adaptive access to an oracle $O_1$ for an $\sf{NEXP}$ language, where the lengths of the oracle queries are polynomial in the input length.
Let $O_1$ be decided by a non-deterministic Turing machine $M_{1}'$ with time complexity $2^{|q|^{k_1}}$, where $k_1$ is a constant and $q$ is the query to the oracle (the input to $M_1'$).
We simulate $M_1^{O_1}$ using another exponential-time oracle Turning machine $M_2$ and another oracle $O_2$, as follows.

Given any input $x$ of length $n$, $M_2$ runs $M_1$ to generate all the oracle queries. For each query $q$,
$M_2$ generates a query $q'$ which is $q$ followed by $2^{|q|^{k_1}}$ bits of $1$.
It then gives all the new queries to its own oracle $O_2$.
Given the oracle's answers, $M_2$ continues running $M_1$ to the end, and accepts if and only if $M_1$ does.
Since $|q|$ is polynomial in $n$, $2^{|q|^{k_1}}$ is exponential in $n$. Furthermore, since there are exponentially many queries and $M_1$ runs in exponential time,
we have that $M_2$ runs in exponential time as well.
It is clear that (1) $M_2$ makes non-adaptive oracle queries, and
(2) $M_2^{O_2}$ decides the same language as $M_1^{O_1}$, as long as $O_2$'s answer to each query $q'$ 
is the same as $O_1$'s answer to the corresponding query $q$.

We define $O_2$ by constructing a non-deterministic Turing machine $M_{2}'$ that simulates $M_{1}'$.
That is, $O_2$ will be the language decided by $M_{2}'$.
More specifically, given a query $q'$ ($q$ followed by $2^{|q|^{k_1}}$ $1$s),
$M_{2}'$ runs $M_{1}'$ on $q$, makes the same non-deterministic choices as $M_{1}'$, and outputs whatever $M_{1}'$ outputs.
Since $M_{1}'$ runs in time $2^{|q|^{k_1}}$,
$M_{2'}$ runs in time polynomial
in its own input size.
Thus, the language $O_2$ decided by $M_{2}'$ is in $\sf{NP}$, and $q' \in O_2$  if and only if $q \in O_1$. 
Accordingly, $M_2^{O_2}$ decides the same language as $M_1^{O_1}$, and we have
$\sf{EXP^{||poly-NEXP}\subseteq EXP^{||NP}}$.

\smallskip
Now, we show $\sf{EXP^{||NP}\subseteq EXP^{||poly-NEXP}}$. The proof is similar to the above.
Let $M_2$ be
an exponential-time oracle Turing machine with non-adaptive access to an oracle $O_2$ for an $\sf{NP}$ language.
Note that the queries made by $M_2$ can be exponentially long.
Let $O_2$ be decided by a non-deterministic Turing machine $M_{2}'$ that runs in time $|q|^{k_2}$, where $k_2$ is a constant and $q$ is the query to $O_2$ (the input to $M_{2}'$).
We simulate $M_2^{O_2}$ using an exponential-time oracle Turning machine $M_1$ and an oracle $O_1$, as follows.

Given any input $x$ of length $n$,
$M_1$ runs $M_2$ to compute the number of oracle queries made by $M_2$, denoted by $Q$.
$M_1$ generates $Q$ oracle queries, with the $i$th query being $x$ followed by the binary representation of $i$.
Since $M_2$ makes at most exponentially many queries, the length of each query made by $M_1$ is (at most) polynomial in $n$.

Query $i$ of $M_1$ is to the following question: {\em is the $i$th query made by $M_2$ given input $x$ in the $\sf{NP}$ language
$O_2$?}
$M_1$ then gives all its queries to its own oracle $O_1$.
Given $O_1$'s answers, $M_1$ uses them to continue running $M_2$, and accepts if and only if $M_2$ does.
Since $M_2$ runs in exponential time, $M_1$ runs in exponential time as well.
It is clear that (1) $M_1$ makes non-adaptive oracle queries, and (2) $M_1^{O_1}$ decides the same language as $M_2^{O_2}$
as long as $O_1$ answers each query correctly.

We define $O_1$ by constructing a non-deterministic Turing machine $M_{1}'$ that simulates $M_{2}'$.
That is, $O_1$ will be the language decided by $M_{1}'$.
More specifically, given an input string of the form $(x, y)$, $M_{1}'$ interprets the second part as the binary representation of an integer $i$.
It runs $M_2$ on $x$ to compute its $i$th query, denoted by $q$. It then runs $M_{2}'$ on $q$,
 makes the same non-deterministic choices as $M_{2}'$, and outputs whatever $M_{2}'$ outputs. Since $q$ is at most exponentially long in $|x|$ and $M_{2}'$ runs in time $|q|^{k_2}$, the running time of $M_{1}'$ is (at most) exponential in its input length.
Thus, the language $O_1$ decided by $M_{O_1}$ is in $\sf{NEXP}$.
Moreover, if $q \in O_2$, then there exist non-deterministic choices that cause $M_{2}'$ and thus $M_{1}'$ to accept;
otherwise both reject.
That is, $O_1$'s answers to the queries by $M_1$ on input $x$ are the same as $O_2$'s answers to the queries by $M_2$ on the same input.

Thus, $M_1^{O_1}$ decides the same language as $M_2^{O_2}$, and we have $\sf{EXP^{||NP}}\subseteq\sf{EXP^{||poly-NEXP}}$.
\end{proof}

\begin{proof}[Proof of \thmref{expchar}]
The theorem
follows immediately from
Lemmas \ref{lem:mriplow}, \ref{lem:mripup},
and \ref{lem:equaloracles}.
\end{proof}

    \section{MRIP Protocols with Two Provers and Constant Rounds}\seclabel{constantrounds}
So far, we allow MRIP protocols to have polynomially many provers and polynomially many rounds, as in MIP protocols in general.
It is well known that any MIP protocol can be simulated using just two provers and
one round~\cite{feige1992two}, which is clearly optimal in terms of both prover number and round number.
In this section, we show similar
results for MRIP protocols.
Recall from Section \ref{subsubsec:23} that we use $\sf{MRIP}[p(n), k(n), t(n)]$ to denote the set of languages that have
MRIP protocols with $p(n)$ provers, $k(n)$ rounds, and $1/t(n)$ utility gap.




\subsection{Constant and Noticeable Utility Gap}\label{subsec:const}
We first prove~\thmref{mrip-const-23} and~\thmref{mrip-poly-23}: that is, any MRIP protocol
with a constant or polynomial utility gap can be simulated by a 2-prover, 3-round MRIP protocol that
retains the corresponding class of utility gaps.
We do so directly using our characterizations in Section \ref{sec:const-polygap}.


%

\begin{proofof}{\thmref{mrip-const-23} \mbox{and} \thmref{mrip-poly-23}}
Recall from~\lemref{notc-lower} and~\lemref{notc-upper} that
\[
    \aMRIP = \sf{P^{||NEXP[\alpha(n)]}},
\]
for any positive integral function $\alpha(n)$ that is polynomially bounded and polynomial-time computable.
We show that 2 provers and 3 rounds are enough to simulate the protocol in~\figref{constantgap}.
Setting $\alpha(n)$ to be a constant or a polynomial in $n$
leads to the corresponding theorems.

More precisely, for any language $L\in \aMRIP$,
we have $L \in \sf{P^{||NEXP[\alpha(n)]}}$.
By definition, there exists a polynomial-time oracle Turing machine $M$ that decides $L$ using $O(\alpha(n))$ non-adaptive queries to an $\sf{NEXP}$ oracle.
Again we assume without loss of generality that the oracle is $\OSAT$ and $M$ makes exactly $\alpha(n)$ oracle queries.
Consider the following 2-prover 3-round variant of the MRIP protocol in \figref{constantgap}  for~$L$.
For any input $x$ of length $n$:
\begin{itemize}[noitemsep,nolistsep,leftmargin=*]
\item
$V$ computes the queries made by $M$, denoted by $\phi_1,\ldots,\phi_{\alpha(n)}$.
%
%
%
%
%
%
%
%
%
%

\item
In the first round,
$P_1$ sends to $V$ the answer bit $c$ to the membership of $x$ in $L$, as well as the answer bits to all
queries, $c_1^*, c_2^*, \ldots, c_{\alpha(n)}^*$,
where $c_i^*$ is the answer to $\phi_i$.
As $P_1$ can compute all oracle queries by running $M$ on $x$,
there is no need for $V$ to send $\phi_1,\dots, \phi_{\alpha(n)}$ to him.

\item
After $V$ has received the answer bits for all $\phi_i$'s, he distinguishes two cases.

For each $i\in \{1, \dots, \alpha(n)\}$ with $c_i^* = 0$, $V$ sets $R_i^*  = 1/2$.
For all $i$'s such that $c_i^* = 1$, $V$
runs the 2-prover 3-round MRIP protocol in \figref{simple-nexp} for the $\phi_i$'s {\em simultaneously}.
That is, for each such $i$, $V$ uses fresh randomness to compute his messages
to $P_1$ and $P_2$ in the second round of the MRIP protocol
for $\phi_i$, denoted by $m_{12}^i$ and $m_{22}^i$ respectively, which are by definition his first messages in the corresponding MIP protocol.
In the second round of the overall protocol, $V$ sends the concatenation of the $m_{12}^i$'s to $P_1$ and the concatenation of the $m_{22}^i$'s to $P_2$.


\item
In the third round, for each $i$ such that $c_i^*=1$,
$P_1$ computes his response $m_{13}^i$
to $m_{12}^i$, and $P_2$ computes his response $m_{23}^i$ to $m_{22}^i$.
They send the concatenation of their responses to $V$.

\item
For each $i$ such that $c_i^*=1$,
$V$ finishes the MIP protocol following the messages exchanged for~$\phi_i$.
If the MIP protocol accepts then $V$ sets $R_i^* =1$; otherwise $R_i^* = 0$.

\item
Finally, $V$ simulates $M$ till the end using the $c_i^*$'s. If the answer bit $c$ does not match $M$'s output, then the protocol ends with $R=-1$;
otherwise the protocol ends with
$R= (\sum_{i=1}^{\alpha(n)} R_i^*)/\alpha(n)$.
$V$ outputs $c$ at the end of the protocol.
\end{itemize}

The correctness of this protocol is similar to \lemref{notc-lower},
except some subtleties caused by the simultaneous execution of the MRIP protocols for the $\phi_i$'s.
First of all, sending $c$ and $c_1^*,\dots, c_{\alpha(n)}^*$ such that the output of $M$ does not match $c$ cannot be part of the provers' best strategy profile,
because it leads to $R=-1$, while sending all messages truthfully leads to $R\geq 1/2$.
Second, by linearity of expectation, for any strategy profile of the provers such that $c$ matches the output of $M$ given $c_1^*,\dots, c_{\alpha(n)}^*$,
the expected payment is the sum of the expected payment for each $\phi_i$.

Note that for each $\phi_i$,  $V$'s messages in the corresponding MIP protocol only depends on his randomness, and he uses fresh coins for $\phi_i$.
Thus,
even though the provers also see $V$'s messages for other $\phi_j$'s, they cannot improve $V$'s marginal accepting probability for $\phi_i$.
From this, the expected payment for each $\phi_i$ is still maximized when the provers report the correct $c_i^*$ and, when $c_i^*=1$,
run the corresponding MIP protocol correctly.
Therefore, under the provers' best strategy profile, the $c_i^*$'s are correct answers to $M$'s oracle queries, $c$ is the correct output of $M$
given the $c_i^*$'s, and $c=1$ if and only if $x\in L$.



Finally, the utility gap of the above protocol is the same as the protocol in  \figref{constantgap}, which is $O(\alpha(n))$.
So we have
$\sf{P^{||NEXP[\alpha(n)]}} \subseteq \sf{MRIP}[2, 3, O(\alpha(n))]\subseteq \aMRIP$,
where the second inclusion is by definition.
Thus we have shown that,
\[\aMRIP = \sf{MRIP}[2, 3, O(\alpha(n))].\]

\thmref{mrip-const-23} holds by setting $\alpha(n)$ to be a constant,
and \thmref{mrip-poly-23} holds
by considering all functions $\alpha(n)=n^k$, where $k\geq 0$ is a constant.
\end{proofof}


\subsection{Negligible Utility Gap}
\label{subsec:neg}

Next, we prove~\thmref{mrip-23}, 
that is,
any MRIP protocol can be simulated by another one using only 2 provers  and 3 rounds.
In the conference version
of this paper~\cite{ChenMcSi16}, we constructed a protocol to simulate
any MRIP protocol using
$2$ provers and $5$ rounds.
In that protocol, the verifier uses $P_1$'s responses to compute his message to~$P_2$, similar to the protocol in \figref{expnexp},
and thus needs 5 rounds.
%
%
%
We left as an open problem whether it is possible to
improve the round complexity to 3, which is the best possible following the discussion at the end of Section~\ref{sec:nexp-using-mip}.


%

In this work, we remove
the dependency between the verifier's messages to the two provers, so they can be sent in parallel,
achieving the optimal round complexity.
Unlike the protocol in Section~\ref{subsec:const}, this simulation does not preserve the utility gap of the original protocol:
even if the latter has a constant or noticeable utility gap, the resulting one has a negligible gap.

\medskip

\begin{proofof}{\thmref{mrip-23}}
Arbitrarily fix an MRIP protocol $(V, \vec{P})$ for a language $L$ with $p(n)$ provers and $k(n)$ rounds.
Without loss of generality, each message in the protocol is of length $\ell(n)$ for any input of length $n$,
where $\ell(n)$ is a polynomial in $n$.
We shift and re-scale the reward function of~$V$, so that 
the payment is always in $[0, 1]$, and 
the expected payment
is strictly larger than 0 under the provers' best strategy profile.
The corresponding 2-prover 3-round protocol $(V', (P_1', P_2'))$ is defined in~\figref{mripinmrip23}.

Essentially, $V'$ asks $P_1'$ to
simulate all provers in the original protocol.
$V'$ wants to use $P_2'$ to
cross-check the transcript provided by $P_1'$, but in parallel: that is, without waiting
for $P_1'$'s message.  He does so by randomly generating a proxy string of
polynomial length and giving it to $P_2'$.  There is an exponentially small
probability that this string is consistent with the transcript $P_1'$ sends, and
if it turns out to be consistent, $V'$ goes on to match the answers he receives from
$P_1'$ and $P_2'$, and to compute the payment as in the 5-round
protocol in~\cite{ChenMcSi16}.


\begin{figure}[tbhp] \centering
\fbox{
\begin{minipage}{0.96\textwidth}
{\normalsize

\noindent
For any input string $x$ of length $n$, the protocol $(V', \vec{P'})$ works as follows:

\begin{enumerate}[leftmargin=15pt]
\item\label{firstround} $P_1'$ sends $m_{11}, \dots, m_{p(n)1}$ to $V'$, where $m_{ij}$ denotes the message sent by prover $P_i$ in round $j$ of $(V, \vec{P})$ 
according to the best strategy profile $s$ of $\vec{P}$.

Let $c$ be the first bit of $m_{11}$. $V'$ outputs $c$ at the end of the protocol.

\item\label{randomstring}

$V'$ generates the random string $r$ used by $V$ and sends it to $P_1'$.
$V'$ selects, uniformly at random, a prover index $i \in \{1,\ldots, p(n)\}$ and a round number $j\in \{2, \ldots, k(n)\}$.
$V'$ then generates a random string ${m}_i^*$ of length $(j-1)\ell(n)$ and sends $(i, j, {m}_i^*)$ to $P_2'$.

\item\label{p2-response}
$P_1'$ uses $r$, $m_{11},\dots, m_{p(n)1}$ and $s$
to continue simulating the protocol $(V, \vec{P})$, and
sends to $V'$ the messages from round~2 to round $k(n)$ in the resulting transcript $\vec{m}$.
$P_2'$ uses $m_i^*$ (and $s$) to
simulate $P_i$ on round $j$, and
sends the resulting message $m'_{ij}$ to $V'$.

\item If ${m}_i^* \neq (m_{i1},\ldots, m_{i(j-1)})$, then the protocol ends with payment $R'=0$.


\item 
If $m_{ij} \neq m'_{ij}$, then $R' = -1$.
Else, $V'$ computes the payment $R$ in the protocol $(V, \vec{P})$ using $x$, $r$ and $\vec{m}$, and sets $R'=\frac{R}{p(n)2^{k(n)\ell(n)}}$.
\end{enumerate}
}
\end{minipage}
}
\caption{Simulating any MRIP protocol with 2 provers and 3 rounds.}\figlabel{mripinmrip23}
\end{figure}

To see why this protocol works, first note that,
even though $V'$ sends to $P_1'$ the randomness $r$ used by~$V$, $V'$ himself uses
fresh randomness
in Step~\ref{randomstring} to generate $i$, $j$ and $m_i^*$, which are unknown to $P_1'$.
Second, the strategy of $P_2'$ in Step~\ref{p2-response} de facto
commits to a strategy profile for the provers in $(V, \vec{P})$ except for the first round,
which together with the randomness $r$ of $V$ and $m_{11},\dots, m_{p(n)1}$ sent by $P_1'$ determines a transcript $\vec{m}^*$ in $(V, \vec{P})$.
%
%

We distinguish two cases for the strategy profiles of $(P_1', P_2')$.

{\bf Case 1.} For some randomness $r$,
$P_1'$ and $P_2'$ do not agree on the transcript under $r$: that is, $\vec{m}\neq \vec{m}^*$,
where $\vec{m}$ is the transcript
sent by $P_1'$. Arbitrarily fix such an $r$.
Suppose $\vec{m}$  disagrees with $\vec{m}^*$ on some $y$ out of $p(n)(k(n)-1)$
messages, with $y\geq 1$.
Then the probability that the prover index $i$ and the round number~$j$ chosen by $V'$
in Step~\ref{randomstring}
satisfy $m_{ij}^*\neq m_{ij}$ is $\frac{y}{p(n)(k(n)-1)}$.

When $m_{ij}^*\neq m_{ij}$, if the random string $m_i^*$ generated  by $V'$ in Step~\ref{randomstring}
does not equal $(m_{i1},\dots, m_{i(j-1)})$, then
the inconsistency between $m_{ij}^*$ and $m_{ij}$ is not caught and the payment is 0;
otherwise the payment is $-1$.
When $m_{ij}^*=m_{ij}$, the payment is either 0 or at most $\frac{1}{p(n)2^{k(n)\ell(n)}}$,
again depending on whether $m_i^* = (m_{i1},\dots, m_{i(j-1)})$ or not.
Finally, as the length of each message in $(V, \vec{P})$ is $\ell(n)$,
for any $i$ and $j$,
the probability that
$m_i^* = (m_{i1},\dots, m_{i(j-1)})$ is $\frac{1}{2^{(j-1)\ell(n)}}\geq \frac{1}{2^{(k(n)-1)\ell(n)}}$.
We upper bound the expected payment $R'$ in Case 1 under $r$ as follows.

\begin{align*}
R' &\leq \sum_{i\leq p(n), 2\leq j\leq k(n)} \frac{1}{p(n)(k(n)-1)}\cdot \frac{1}{2^{(j-1)\ell(n)}} \cdot \left( \mathbb{I}_{m_{ij}^*\neq m_{ij}}\cdot  (-1) +  \mathbb{I}_{m_{ij}^*= m_{ij}} \cdot  \frac{1}{p(n)2^{k(n)\ell(n)}}\right) \\
&\leq - \frac{y}{p(n)(k(n)-1)}\cdot \frac{1}{2^{(k(n)-1)\ell(n)}} \\
& \quad + \sum_{i\leq p(n), 2\leq j\leq k(n)} \frac{1}{p(n)(k(n)-1)}\cdot \frac{1}{2^{(j-1)\ell(n)}} \cdot  \mathbb{I}_{m_{ij}^*= m_{ij}} \cdot  \frac{1}{p(n)2^{k(n)\ell(n)}}\\
&< - \frac{y}{p(n)(k(n)-1)}\cdot \frac{1}{2^{(k(n)-1)\ell(n)}} + \sum_{2\leq j\leq k(n)} \frac{1}{k(n)-1} \cdot \frac{1}{2^{(j-1)\ell(n)}} \cdot  \frac{1}{p(n)2^{k(n)\ell(n)}} \\
&< - \frac{y}{p(n)(k(n)-1)}\cdot \frac{1}{2^{(k(n)-1)\ell(n)}} + \frac{1}{(k(n)-1)p(n)2^{k(n)\ell(n)}} \\
&= \frac{1-2y}{(k(n)-1)p(n)2^{k(n)\ell(n)}} < 0.
\end{align*}

On the other hand, if $P_1'$ acts consistently with $P_2'$ in Step~\ref{p2-response}  under $r$, and keeps his strategy unchanged under any other randomness of $V$ sent to him by $V'$, then
the expected payment under $r$ is at least 0 and the expected payment under any other randomness of $V$ does not change;
therefore, the expected payment in the whole protocol gets larger.
Accordingly, under the best strategy profile of $(P_1', P_2')$, Case 1 does not occur for any randomness $r$ of $V$.

\smallskip

{\bf Case 2.}
In their strategy profile $s'$, $P_1'$ and $P_2'$ agree on the
transcript $\vec{m}$ under every randomness $r$ of $V$, 
but the strategy profile $\tilde{s}$ committed by them for $(V, \vec{P})$ (that is, by
$P_1'$ in Step~\ref{firstround} for round 1 and then by $P_2'$ in Step~\ref{p2-response} for the remaining rounds)
has 
the answer bit $c$ incorrect. Thus $\tilde{s}$ is not the best strategy profile $s$ of $\vec{P}$.

In this case, given any randomness $r$, prover $i$ and round $j$ chosen by $V'$ in  Step~\ref{randomstring}, 
the expected payment is
\[
 R' = \frac{1}{2^{(j-1)\ell(n)}} \cdot \frac{R}{p(n)2^{k(n)\ell(n)}},
\]
where $R$ is the payment of $(V, \vec{P})$ under $\tilde{s}$ and $r$.
Therefore, the expected payment for $P_1'$ and $P_2'$ in the whole protocol is
\begin{align*}
u_{(V', \vec{P}')}(s'; x) & = \sum_{i\leq p(n), 2\leq j\leq k(n)} \frac{1}{p(n)(k(n)-1)} \cdot \frac{1}{2^{(j-1)\ell(n)}} \cdot \frac{u_{(V, \vec{P})}(\tilde{s};x)}{p(n)2^{k(n)\ell(n)}} \\
& < \sum_{i\leq p(n), 2\leq j\leq k(n)} \frac{1}{p(n)(k(n)-1)} \cdot \frac{1}{2^{(j-1)\ell(n)}} \cdot \frac{u_{(V, \vec{P})}(s;x)}{p(n)2^{k(n)\ell(n)}},
\end{align*}
where the inequality is because $u_{(V, \vec{P})}(\tilde{s};x)< u_{(V, \vec{P})}(s;x)$.
Note that the second line in the equation above is exactly the expected payment for $P_1'$ and $P_2'$ when they commit to $s$.
Thus committing to $\tilde{s}$ is not the best strategy profile for $P_1'$ and $P_2'$.

\smallskip
In sum, a best strategy profile for the provers in $(V', \vec{P}')$ is to commit to a best strategy profile $s$ in $(V, \vec{P})$,
and the corresponding answer bit $c$ is 1 if and only if $x\in L$, following fact that $(V, \vec{P})$ is an MRIP protocol for $L$.
\end{proofof}

\section*{Acknowledgments} \seclabel{ack}
We thank anonymous reviewers for their valuable feedback that
helped improve this paper, and Sanjoy Das, Andrew Drucker, Silvio Micali and Rafael Pass for helpful comments.
This work has been partially supported by NSF CAREER Award CCF 1553385,
CNS 1408695, CCF 1439084, IIS 1247726, IIS 1251137, and CCF 1217708, and Sandia National Laboratories.

\bibliographystyle{plain}


\end{document}